\documentclass[12pt, leqno]{amsart}

\usepackage{mathtools}
\usepackage{graphicx}
\usepackage{amssymb,amsmath,amsthm,amscd}
\usepackage{mathrsfs}
\usepackage{lscape}
\usepackage{enumerate}
\usepackage[usenames,dvipsnames]{color}
\usepackage[colorlinks=true, pdfstartview=FitV, linkcolor=blue,citecolor=blue,urlcolor=blue]{hyperref}
\usepackage{comment}
\usepackage{tikz}
\usepackage[all]{xy}

\allowdisplaybreaks[3]

\numberwithin{equation}{section}

\newenvironment{red}{\relax\color{red}}{\relax}
\newenvironment{blue}{\relax\color{blue}}{\hspace*{.5ex}\relax}

\newcommand{\ber}{\begin{red}}
\newcommand{\er}{\end{red}}
\newcommand{\beb}{\begin{blue}}
\newcommand{\eb}{\end{blue}}

\theoremstyle{plain}
\newtheorem{lemma}{Lemma}[section]
\newtheorem{proposition}[lemma]{Proposition}
\newtheorem{theorem}[lemma]{Theorem}

\newtheorem{corollary}[lemma]{Corollary}

\theoremstyle{definition}
\newtheorem{remark}[lemma]{Remark}
\newtheorem{example}[lemma]{Example}

\newtheorem{definition}[lemma]{Definition}

\newcommand{\n}{\mathfrak{n}}
\newcommand{\g}{\mathfrak{g}}

\newcommand{\C}{\mathbb{C}}

\newcommand{\Z}{\mathbb{Z}}

\newlength{\mylength}
\title{Supersymmetric bi-Hamiltonian systems}

\author{Sylvain Carpentier}
\address[S. Carpentier]
{Department of Mathematics, Columbia University in the City of New York, Broadway and 116th, New York, NY 10027 }
\email{sc.4278@columbia.edu}
\author{Uhi Rinn Suh}
\address[U.R. Suh]
{ Department of Mathematical Sciences and Research institute of Mathematics, Seoul National University, GwanAkRo 1, Gwanak-Gu, Seoul 08826,
Korea}
\email{uhrisu1@snu.ac.kr}

\date{\today}

\begin{document}

\begin{abstract}
We construct super Hamiltonian integrable systems within the theory of Supersymmetric Poisson vertex algebras (SUSY PVAs).  We provide a powerful tool for the understanding of SUSY PVAs called the super master formula.  We attach some Lie superalgebraic data to a generalized SUSY W-algebra and show that it is equipped with two compatible SUSY PVA brackets. We reformulate these brackets in terms of odd differential operators and obtain super bi-Hamiltonian hierarchies after performing a supersymmetric analog of the Drinfeld-Sokolov reduction on these operators. As an example, an integrable system is constructed from $\g=\mathfrak{osp}(2|2)$.
\end{abstract}

\maketitle

\section{Introduction}

 In \cite{DS}, Drinfeld and Sokolov introduced a new local Poisson algebra later called a classical W-algebra. For a  finite semisimple Lie algebra $\g$ and its principle nilpotent element $f_{\text{pr}}$, the classical W-algebra $W(\g, f_{\text{pr}})$
 is obtained by Hamiltonian reduction from the Poisson algebra $S(\hat{\g})$. These algebras have been studied both in physics and mathematics \cite{BP93}.
 In the same paper \cite{DS}, Drinfeld and Sokolov constructed a hierarchy of integrable bi-Hamiltonian systems on $W(\g, f_{\text{pr}})$ by considering two compatible Lie brackets on its quotient space $W(\g, f_{\text{pr}})/\partial W(\g, f_{\text{pr}})$. 
  Ever since, the so-called generalized Drinfeld-Sokolov hierarchies have flourished in the literature, see the papers \cite{BDHM, DFG98, DHM, DKV13, DKV14, FHM93, FM97, Ped95, LWZ11, S18, W17} and references there.

\vskip 2mm

In the recent articles  \cite{DKV13, DKV14}  by De Sole, Kac, Valeri and the second author \cite{S18},  generalized Drinfeld-Sokolov hierarchies have been tackled within the theory of Poisson vertex algebras
(PVAs). This framework introduced in \cite{BDK} proved to be successful as the authors of \cite{DKV18,  DKV16, DKV15,DKV16_2, DKV15_2}  were able to construct integrable biHamiltonian hierarchies on the classical W-algebra $W(\g,f)$, for any semisimple Lie algebra $\g$ of classical type and any nilpotent element $f \in \g$. Explicit formulas for these equations and their solutions were found in \cite{CDKVV19} by de Sole, Kac, Valeri, van de Leur and the first author in the case $\g=\mathfrak{gl}_N$.

\vskip 2mm

In this paper, we investigate super Hamiltonian equations in terms of Supersymmetric Poisson Vertex Algebras (SUSY PVAs), invented in \cite{HK06}. A classic and well-known example of such equations is the super KdV equation (see \cite{Mc92} for instance)
\begin{equation} \label{sKdV-6}
\psi_t=\psi^{(6)} + 3 \psi'' \psi' + 3\psi \psi^{(3)}, \, \, \psi \, \,  \text{odd}, \, \, \psi^{(n)}:=D^n(\psi),
\end{equation}
who can be cast in a super Hamiltonian form using the (local) Neveu-Schwarz SUSY PVA bracket
\[ \{\psi{}_\chi\psi\}= \textstyle  (D^2-\frac{3}{2}\chi^2+ \frac{1}{2}\chi D)\psi -\chi^5\]
as follows 
\begin{equation} \label{sKdV-2}
\psi_t=\textstyle\{  \psi \psi' {}_\chi \psi\}|_{\chi=0}.
\end{equation}

\vskip 2mm
Super Hamiltonian equations have been studied in various articles \cite{DG98,  IK92,  IK91, I91,  IK91_2, Kuper,  KZ, ManRa85,Mc92, OP90}.  Oevel and Popowicz discovered in \cite{OP90} that the super KdV equation \eqref{sKdV-6} is  super Hamiltonian for a second nonlocal bracket. However in this paper we are exclusively interested in systems that are bi-Hamiltonian for two local compatible brackets. We plan on studying nonlocal SUSY PVAs in a future publication.
\vskip 2 mm

 In \cite{IK92, IK91}, Inami and Kanno constructed integrable systems associated to  a family of odd differential operators taking values in the affine Lie superalgebras $\mathfrak{sl}(n|n)^{(1)}$ and $\mathfrak{osp}(2|2)^{(2)}$. 
Our overarching goal in this paper is to extend their method to a much larger class of systems and to prove that these systems are bi-Hamiltonian in the framework of local SUSY PVAs. This generalization is complementary to the approach of Delduc and Gallot in \cite{DG98}.

\if

In Section \ref{Sec:PVA and HIS}, we review the notions in PVAs theory and Hamiltonian integrable systems via $\lambda$-brackets. If we know a biHamiltonian structure of an integrable system, the integrability can be proved  via Lenard-Magri scheme. At the end of the section, we provide the statement of Lenard-Magri scheme.

\fi

\vskip 2mm
Simply speaking, a SUSY PVA structure on a differential superalgebra $\mathcal{P}$ with an odd derivation $D$ is a linear map $ \mathcal{P} \otimes_{\C} \mathcal{P} \rightarrow \mathcal{P}[\chi]$, where $\chi$ is an odd indeterminate, satisfying the sesquilinearity, skew-symmetry, Jacobi identity and Leibniz rule axioms (see Definition \ref{Def:SPVA}).
We prove in Section \ref{Sec: SUSY PVA and sHIS} that, when $\mathcal{P}$ is a superalgebra of differential polynomials, these axioms hold if and only if they hold on its generators. We give in this case an explicit formula, called the \textit{super master formula} \eqref{Eqn:Master formula2}, for the $\chi$-bracket of any two elements in $\mathcal{P}$ in terms of the $\chi$-brackets of its generators.

\vskip 2mm
The induced bracket $\{ \smallint a, \smallint b\}:=\smallint \{a _\chi b \}_{\chi=0}$ on the quotient space $\mathcal{P}/D \mathcal{P}$ is a Lie superalgebra of degree $\bar{1}$ bracket. Moreover, $\mathcal{P}/D \mathcal{P}$ naturally acts by derivations on $\mathcal{P}$ itself as follows: $(\smallint a).b :=\{ \smallint a _\chi b \}_{\chi=0}$. In this language, the super Hamiltonian equation associated to the Hamiltonian $\smallint h$ is the system 
\[
\frac{du_i}{dt}=\{\smallint h _{\chi} u_i \}_{\chi=0}, \, i \in I, \, \, \, \, \text{where }  \mathcal{P}=\C[D] \otimes \text{Span}_{\C} \{u_i | i \in I \}.
\]
A (super) integrable system on $(\mathcal{P}, \{ \ {}_\chi \ \})$ is an infinite dimensional abelian subalgebra of $(\mathcal{P}/D \mathcal{P}, \{ , \})$ or, equivalently, a family of pairwise compatible super Hamiltonian equations whose Hamiltonians form an infinite dimensional subspace of $\mathcal{P}/D \mathcal{P}$.

\vskip 2 mm

Let $\g$ be a simple Lie superalgebra. Consider a nondegenerate invariant supersymmetric bilinear form $( \, | \, )$ and  an odd element $s_0$ in $\g$. The superalgebra $\mathcal{P}(\overline \g)$ of differential polynomials on the parity reversed vector superspace $\overline \g$ admits two compatible SUSY PVA brackets, which are defined on its generators as follows:
\begin{equation}
\begin{split}
\{\bar{a} {}_\chi \bar{b} \}_1&=(-1)^{p(a)}( \overline{ [a,b]} +\chi (a|b)),\\
\{\bar{a} {}_\chi \bar{b} \}_2&=(-1)^{p(a)+1}(s_0|[a,b]).
\end{split}
\end{equation}
$\mathcal{P}( \overline \g)$ is called the affine SUSY PVA associated to $\g$.
We pick an arbitrary, not necessarily principal, $\Z$-grading of $\g$
\begin{equation}
\g=\bigoplus_{ k \in \Z} \g_k.
\end{equation}
Let $\n \subset \g_{>0}$ be a subalgebra of $\g$, $\mathfrak{m} \subset \n$  be an $\n$-module and let $f \in \g_{-i}, s \in \g_j, i,j >0$ be two odd elements such that
\begin{enumerate}
\item[(A-1)] $\text{ad} f|_{\n}$ is injective,
\item[(A-2)]$(f|[\n, \mathfrak{m}])=0$,
\item[(A-3)] $\g_{ \geq i} \subset \mathfrak{m}$,
\item[(A-4)] $[s, \n]=0$,
\item[(A-5)] $\text{ad} \, \, \Lambda^2$ is semisimple in $\hat \g=\g(\!(z^{-1})\!)$ where $\Lambda:=f+zs$.
\end{enumerate}
Let $\mathfrak{b}$ be the orthogonal of $\mathfrak{m}$ for the bilinear form and $\mathfrak{b}_-$ be a complement of $\mathfrak{m}$ in $\g$. The key object of our construction is the odd differential operator 
\[ \mathcal{L}(\mathcal{Q})= D+\mathcal{Q} + \Lambda\otimes 1 \in \C[D]\ltimes \g(\!( z^{-1})\!) \otimes \mathcal{P},  \]
where $\mathcal{P}$ is the superalgebra  of differential polynomials generated by $\overline{\mathfrak{b}}_-$  and $\mathcal{Q} \in \mathfrak{b} \otimes \mathcal{P}$. The Lie algebra $(\n \otimes \mathcal{P})_{\bar{0}}$ acts by gauge transformations on this class of operators, which induces a well-defined $\n$-action 
 on the superalgebra $\mathcal{P}$. We show in Section \ref{Sec: DSred} that the set of invariant elements in $\mathcal{P}$ forms a superalgebra
 $\mathcal{W}(\g,\n, \mathfrak{m},f)$, called the \emph{generalized W-algebra associated to} $\g, \n, \mathfrak{m}$ and $f$. Moreover,  $\mathcal{W}(\g,\n, \mathfrak{m},f)$  admits two compatible SUSY PVA brackets. We give an equivalent construction of these brackets as a reduction of the two compatible brackets on the affine SUSY PVA associated to $\g$. The generalized W-algebra is generated by $\text{dim} \,  \g-\text{dim}\, \n-\text{dim} \,  \mathfrak{m}$ elements.
\vskip 2 mm

In Section \ref{Sec:DS hier}, we construct bi-Hamiltonian hierarchies attached to our superalgebraic data. To each element $C$ in the center $\mathcal{Z}$ of $\text{ker} \, \text{ad} \, \Lambda^2$ we associate a functional $\smallint h_C$ in $\mathcal{W}(\g,\n, \mathfrak{m},f)/D \mathcal{W}(\g,\n, \mathfrak{m},f)$ and an element $M_C \in \hat{\g} \otimes \mathcal{P}$ which supercommutes with the \textit{universal Lax operator} 
\begin{equation}
\mathcal{L}:=\mathcal{L}(\mathcal{Q}_u), \, \, \mathcal{Q}_u=\sum_{t \in I} q^t \otimes \bar{q}_t, \, \, (q^t|q_{t'})=\delta_{t,t'}
\end{equation}
where $(q^t)_{t \in I}$ ($\text{resp.} (q_t)_{t \in I}$) is a basis of $\mathfrak{b}$ (resp. $\mathfrak{b}_-$).
Our main results are the followings:

\begin{theorem}\label{theintro}
The equations
\begin{equation} \label{toprove}
\frac{d \mathcal{L}}{dt_C}=[M_{C}^+,\mathcal{L}], \, \, C \in \mathcal{Z}
\end{equation}
are well-defined on $\mathfrak{b} \otimes \mathcal{P}$ and pairwise compatible. They induce a hierarchy of pairwise compatible super Hamiltonian equations on $\mathcal{P}$
\begin{equation}
\frac{d \bar{q}_t}{dt_C}=\{ \smallint h_C {}_\chi \bar{q}_t \}_2 |_{\chi=0},\, \, t \in I.
\end{equation}
\end{theorem}
Moreover, we show that these equations induce a bi-Hamiltonian hierarchy on $\mathcal{W}(\g, \n, \mathfrak{m},f)$.

\begin{theorem}
 The super Hamiltonian equations associated to elements $C \in \mathcal{Z}$
\begin{equation}\label{redSHEintro}
\frac{d \phi_k}{dt_C} = \{\smallint h_{C}{}_\chi  \phi_k \}_1^{\mathcal{W}}|_{\chi=0}, \, \, \, \, k=1,..., r
\end{equation} 
where $\mathcal{W}(\g, \n, \mathfrak{m},f)=S(\C[D] \otimes \text{Span}_{\C}\{ \phi_1,..., \phi_r \}$, are pairwise compatible. In particular, the Hamiltonians $(\int h_{D})_{D \in \mathcal{Z}}$ are integrals of motion for each of these equations. Moreover, \eqref{redSHEintro} can be rewritten using the second reduced bracket:
\begin{equation*}
\frac{d \phi_k}{dt_C} = \{\smallint h_{zC}{}_\chi \phi_k \}_2^{\mathcal{W}}|_{\chi=0}, \, \, \, \, k=1,...,r.
\end{equation*}
\end{theorem}

 Finally, we apply in Section \ref{Sec:example} our results to the simple Lie superalgebra $\g=\mathfrak{osp}(2|2)$. In particular, we find a bi-Hamiltonian hierarchy on two even variables $w_1, w_3$ and two odd variables $w_2,w_4$ whose first equation is 
\begin{equation}\label{ex_eq1 intro}
\begin{split}
\frac{dw_1}{dt_1}&=w_1''-2(w_1w_2)'-2w_1w_3, \, \, \, \frac{dw_2}{dt_1}=0, \\
\frac{dw_3}{dt_1}&=-8 w_1 w_2 w_4, \, \, \,\frac{dw_4}{dt_1}=w_4''+2w_2w_4'+2w_3w_4.
\end{split}
\end{equation}

\vskip 2 mm
\textit{Acknowlegments}. The first
author was supported by a Junior Fellow award from the Simons Foundation. He is extremely grateful to the Seoul National University for its hospitality during a short visit in the fall of 2019, where most of the work for this paper was accomplished. 
The  second author was supported by the New Faculty Startup Fund from Seoul National University and the NRF Grant \# NRF-2019R1F1A1059363.

\section{Supersymmetric Poisson vertex algebras and Hamiltonian Systems} \label{Sec: SUSY PVA and sHIS}

\subsection{Supersymmetric poisson vertex algebras}

In this section, we review the theory of Supersymmetric Poisson Vertex Algebras (SUSY PVAs), who first appeared in \cite{HK06}. We introduce and derive the super master formula  \eqref{Eqn:Master formula2} in the case of superalgebras of differential polynomials.

\begin{definition} \hfill
\begin{enumerate}
\item A {\it vector superspace} $V=V_{\bar{0}}\oplus V_{\bar{1}}$ is a $\Z/2\Z$-graded vector space. An element in $V_{\bar{0}}$ (resp. $V_{\bar{1}}$) is called {\it even} (resp. {\it odd}). The parity of a homogeneous element $a$ is defined by 
$ p(a)= \left\{ \begin{array}{ll} 0  & \text{if $a$ is even}, \\ 1 & \text{if $a$ is odd.} \end{array} \right.$
\item A {\it superalgebra} A is a vector superspace with a  binary operation $A\otimes A \to A$ satisfying 
$ A_{\bar{i}} A_{\bar{j}} \subset A_{\bar{i}+\bar{j}}$ for $\bar{i}, \bar{j}\in \Z/2\Z.$
\item A {\it Lie superalgebra} $\g$ is a superalgebra endowed with the binary operation called Lie bracket  $[\ , \ ]: \g\otimes \g \to \g$ satisfying the following properties:
\begin{itemize}
\item(skewsymmetry) \quad $[a,b]= (-1)^{p(a)p(b)+1}[b,a],$
\item(Jacobi identity) \quad $[a,[b,c]]= [[a,b],c]+(-1)^{p(a)p(b)} [b,[a,c]].$
\end{itemize}
\end{enumerate}
\end{definition}

 In the sequel, a  
$\mathbb{C}[D]$-module is a vector superspace on which $D$ acts as an odd derivation.

\begin{definition}  \label{Def:SLCA}
A $\chi$-bracket on a $\mathbb{C}[D]$-module $\mathcal{R}$ is a linear map  
\[ \, [ \ _\chi \ ] : \mathcal{R} \otimes_\C \mathcal{R} \to \mathcal{R}[ \chi] \]
\[  a \otimes b \mapsto    \sum_{n \geq 0} \chi^{n}  a_{(n)}b  \]
such that $ p(a_{(n)}b)=p(a)+p(b)+n+1$ for all $n \geq 0$ and $a, b \in \mathcal{R}$. $\chi$ is an odd indeterminate and  we give  the structure of a $\mathbb{C}[D]$-module to $\mathcal{R}[ \chi]$ via the equation
\begin{equation}
D \chi+\chi D=-2 \chi^2.
\end{equation}

A $\mathbb{C}[D]$-module $\mathcal{R}$ endowed with a $\chi$-bracket is called a {\it ($N_k=1$) supersymmetric Lie conformal algebra (SUSY LCA)} if it satisfies the following properties for any $a,b,c \in \mathcal{R}$:
\begin{itemize}
\item (sesquilinearlity) \quad  $[Da_\chi b]= \chi[a_\chi b], \quad [a_\chi Db]= (-1)^{p(a)+1} (D+\chi)[a_\chi b]$.
\item (skew-symmetry) \quad $[a_\chi b]= (-1)^{p(a)p(b)} [b_{-\chi-D}a],$ where 
\[\textstyle [b_{-\chi-D} \, a]=\sum_{n \geq 0} (-D-\chi)^{n} b_{(n)}a.\]
\item (Jacobi-identity) \quad $[a_\chi [b_\gamma c]]=(-1)^{p(a)+1} [[a_\chi b]_{\chi+\gamma} c]+(-1)^{(p(a)+1)(p(b)+1)}[b_\gamma[a_\chi c]]$. 
\end{itemize}
$\gamma$ is an odd indeterminate which supercommutes with $\chi$. The LHS and the second term of the RHS in the Jacobi identity can be computed using the formulas
\begin{equation}
\begin{aligned}
& \textstyle [a_\chi \gamma b]= (-1)^{p(a)+1}\gamma [a_\chi b] \\
& \textstyle [a_{\, \gamma} \chi b]= (-1)^{p(a)+1}\chi [a_{\, \gamma} b]
\end{aligned}
\end{equation}
for any $a,b \in \mathcal{R}$. As for the first term of the RHS, we use
\begin{equation}
[ \chi a_{\chi+\gamma}b]=- \chi [a_{\chi+\gamma}b].
\end{equation}
\end{definition}

\begin{definition}
We denote by $\smallint a$  the projection of an element $a \in \mathcal{R}$ onto $\mathcal{R}/D \mathcal{R}$ and call elements of $\mathcal{R}/D \mathcal{R}$ \textit{functionals}.
\end{definition}
\begin{lemma} \cite{HK06}
Let $(\mathcal{R}, [ \,\,  {}_{\chi} \, \, ],D)$ be a SUSY LCA. The following bracket is a Lie superalgebra bracket of degree $\bar{1}$ on $\mathcal{R}/D \mathcal{R}$:
\begin{equation}\label{indbra}
[ \smallint a, \smallint b]:=\smallint [a_{\chi} b ]|_{\chi=0}, \, \, a,b \in \mathcal{R}.
\end{equation}
Namely, we have for all $\smallint a, \smallint b \in \mathcal{R}/D \mathcal{R}$,
\begin{equation*}
\begin{split}
& [\smallint a, \smallint b ]=(-1)^{p(a)p(b)} [ \smallint b, \smallint a ],\\
& [ \smallint a, [ \smallint b , \smallint c ] ]=(-1)^{p(a)+1} [ [ \smallint a, \smallint b ], \smallint c ]+(-1)^{(p(a)+1)(p(b)+1)} [\smallint b, [ \smallint a, \smallint c ] ].
\end{split}
\end{equation*}
\end{lemma}
\begin{proof}
Immediate from the axioms of a SUSY LCA.
\end{proof}
\begin{definition} \label{Def:SPVA}
A tuple $(\mathcal{P}, \{\, _\chi \, \},1, \cdot, D)$ is called a {\it ($N_k=1$) supersymmetric Poisson vertex algebra (SUSY PVA)}  if 
\begin{itemize}
\item $(\mathcal{P}, \{\, _\chi \, \}, D)$ is a SUSY LCA,
\item $(\mathcal{P}, 1, \cdot, D)$ is a unital associative supercommutative superalgebra endowed with an odd derivation $D$. 
\item (right) Leibniz rule: $\{a_\chi bc\}= \{a_\chi b\} c +(-1)^{p(b)p(c)}\{a_\chi c\} b$ for all $a,b,c \in \mathcal{P}.$
\end{itemize}
\end{definition}
\begin{remark}
A SUSY PVA $(\mathcal{P}, \{\, _\chi \, \},1, \cdot, D)$ has a natural PVA structure where  $\partial= D^2$, $\lambda=-\chi^2$ and 
\begin{equation} \label{chi to lambda bracket}
\textstyle  \{a{}_\lambda b\}= \sum_{n \geq 0} (-\lambda)^n a_{(2n+1)}b, \text{ for } a,b\in \mathcal{P}.
\end{equation}
 Indeed, one can prove that \eqref{chi to lambda bracket} satisfies each axiom of a PVA using the corresponding axiom of a SUSY PVA. We refer to \cite{BK,DK} for the $\lambda$-bracket formalism of VAs and PVAs.
\end{remark}

We are going to show that the (right) Leibniz rule and the skew-symmetry axioms imply the left Leibniz rule \eqref{leftleibniz}. In order to do so we begin with an auxiliary Lemma.

\begin{lemma} \label{Lem:LLR_2}
Let $\mathcal{P}$ be a SUSY PVA. Then, for any $a,b,c \in \mathcal{P}$ we have
\[ \{a_{\chi+D}b\}_\to c=(-1)^{p(a)p(b)}\{b_{-\chi-D}a\}c \]
where
$\{a_{\, \chi+D} \, b\}_{\to}:= \sum_{n \geq 0} (-1)^{n\cdot (p(a)+p(b))} a_{(n)}b \, (\chi+D)^{n}.$
\end{lemma}

\begin{proof}
By skewsymmetry,  we have for all $a,b,c \in \mathcal{P}$
\begin{equation} \label{sseq}
\sum_{n \geq 0} \chi^n b_{(n)}a . c=(-1)^{p(a)p(b)} \sum_{n \geq 0} (-\chi-D_1)^n a_{(n)}b. c \, .
\end{equation}
We write $D_1$ to emphasize that it only acts as $D$ on $a_{(n)}b$, not on $c$. We can replace $\chi$ by $-\chi-D_2$ in \eqref{sseq}
\begin{equation} \label{ssseq}
(-\chi-D_2)^n b_{(n)}a . c=(-1)^{p(a)p(b)}(\chi+D_2-D_1)^n a_{(n)}b c
\end{equation}
where $D_2$ acts as $D$ on everything on its right in \eqref{ssseq}. Therefore we have
\begin{equation} \label{seq}
\sum_{n \geq 0} (-\chi-D)^n(b_{(n)}a. \,c)=(-1)^{p(a)p(b)+np(ab)}a_{(n)}b (\chi+D)^n c,
\end{equation} 
proving the claim.
\end{proof}

We are now ready to prove that a SUSY PVA satisfies the left Leibniz rule.
\begin{lemma}\label{Lem:left leibniz}
Let $\mathcal{P}$ be a SUSY PVA. Then, for any $a,b,c \in \mathcal{P}$ we have
\begin{equation} \label{leftleibniz}
 \{ab_{\, \chi }\, c\}= (-1)^{p(b)p(c)} \{a_{\, \chi+D} \, c\}_\to b+(-1)^{p(a)p(bc)}\{b_{\, \chi+D} \, c\}_\to a.
 \end{equation}
\end{lemma}

\begin{proof}
By sesquilinearity, Lemma \ref{Lem:LLR_2} and the right Leibniz rule,
\begin{equation*}
\begin{aligned}
 \{ab_{\, \chi }\, c\}& =(-1)^{p(c)p(ab)}\{c_{\, -\chi-D} \, ab\} \\
 & =  (-1)^{p(a)p(c)+p(b)p(c)}  \{c_{-\chi-D}a\}b+(-1)^{p(a)p(bc)+p(b)p(c)} \{c_{-\chi-D}b\}a\\
 &=  (-1)^{p(b)p(c)} \{a_{\, \chi+D} \, b\}_\to c + (-1)^{p(a)p(bc)}\{b_{\, \chi+D} \, c\}_\to a.
 \end{aligned}
 \end{equation*}
\end{proof}

\begin{definition} \label{defXf}
Let $\mathcal{P}$ be a SUSY PVA. It follows immediately from the sesquilinearity axiom and the Leibniz rule that for all $ a \in \mathcal{P}$, the map $b \rightarrow \{  a _\chi b \}|_{ \chi=0}$ is a derivation of $\mathcal{P}$ of parity $p(a)+1$ which supercommutes with $D$. We denote this derivation by $X_{\smallint a}$ since it only depends on the class $\smallint a$ by sesquilinearity.
\end{definition}

The main corollary of the following Lemma is that a pair of functionals in involution $\{ \smallint a , \smallint b \}=0$ defines a pair of supercommuting derivations of the superalgebra $\mathcal{P}$.
\begin{lemma} \label{eqcom}
For all $ a, b \in \mathcal{P}$, we have
\begin{equation} \label{eqncom}
[X_{\smallint a}, X_{\smallint b}]=(-1)^{p(a)+1}X_{ \{ \smallint a, \smallint b\}}.
\end{equation}
\end{lemma}
\begin{proof}
It follows immediately from the Jacobi identity and the definition \eqref{indbra} of the induced bracket on $\mathcal{P}/D \mathcal{P}$.
\end{proof}

From now on, $V$ denotes a vector superspace with basis $\mathcal{B}=\{u_i, i \in I\}$. We assume that $I \subset \Z$ and that the parity of $u_i$ is the same as the parity of the integer $i$. Let $\mathcal{R}=\mathbb{C}[D] \otimes V$ be the $\mathbb{C}[D]$-module freely generated by $V$ and $\mathcal{P}=S(\mathcal{R})$ the superalgebra of super differential polynomials in the $u_i$'s. As we will see in \eqref{Eqn:Master formula2}, a SUSY PVA structure on $\mathcal{P}$ is completely determined by the $\chi$-brackets $\{ u_i {}_\chi u_j \}$ between pairs of basis elements $(u_i,u_j), i, j \in I$.

\begin{definition} \label{Def:pdev}
The partial derivative with respect to the variable $u_i^{(m)}:=D^m(u_i)$ is the unique derivation $\frac{\partial}{\partial u_i^{(m)}}$ of $\mathcal{P}$ of parity $i+m$ such that $\frac{\partial u_j^{(n)}}{\partial u_i^{(m)}}=\delta_{i,j} \delta_{m,n}$ for all $i \in I, n \geq 0$.
We will simply denote $\frac{\partial a}{\partial u_i^{(m)}}$ by $a_{(i,m)}$.
\end{definition}

\begin{lemma} \label{Lemma: commutator_derivatives}
For all $i \in I$ and $m \geq 0$ we have
\begin{equation} \label{Eqn: commutator_derivatives}
 \bigg[ \frac{\partial}{\partial u_i^{(m)}}, D \bigg] = \frac{\partial}{\partial u_i^{(m-1)}}, 
\end{equation}
where by definition $ \frac{\partial}{\partial u_i^{(-1)}}=0$.
\end{lemma}

\begin{proof}
One can easily check that, for all $j \in I$ and $n \geq 0$ 
\begin{equation} \label{eqqqq}
\textstyle [ \frac{\partial}{\partial u_i^{(m)}}, D ](u_j^{(n)}) = \frac{\partial}{\partial u_i^{(m-1)}}(u_j^{(n)})=\delta_{i,j} \delta_{n, m-1}.
\end{equation} 
Therefore equation \eqref{Eqn: commutator_derivatives} holds since $\frac{\partial}{\partial u_i^{(m-1)}} $ is by definition the unique derivation of parity $i+m-1$ satisfying \eqref{eqqqq}.
\end{proof}

\begin{proposition} \label{Prop:Master formula}
Suppose that $\mathcal{P}$ is a SUSY PVA with bracket $\{ \, _\chi \, \}$. Then 
\begin{equation} \label{Eqn:Master formula2}
\begin{aligned}
 \{ a_\chi b\} = \sum_{i,j\in I; m,n \geq 0} &  (-1)^{S(a_{(i,m)},b_{(j,n)})} (-1)^{n(i+m+1)+m(i+j+1)+ \frac{(m-1)m}{2}} \\
 & b_{(j,n)}(\chi+D)^n\{u_i { }_{\, \chi+D} \ u_j \}_{\to}(\chi+D)^m a_{(i,m)}.
 \end{aligned}
  \end{equation}
for all $a,b \in \mathcal{P}$, where 
 $S(a_{(i,m)},b_{(j,m)}):=p(b)(p(a)+1)+(p(b)+i+m)(j+n).$
\end{proposition}
\begin{proof}
The proposition directly follows from the Leibniz rule, equation \eqref{leftleibniz} and the sesquilinearity of a SUSY PVA.
\end{proof}

The equation \eqref{Eqn:Master formula2} is the supersymmetric analog of the master formula (\cite{BDK}). We will refer to it as the \textit{super master formula}.

\begin{remark}
If $\mathcal{R}$ is a SUSY LCA, the super master formula \eqref{Eqn:Master formula2} is equivalent by sesquilinearity to
\begin{equation}\label{Eqn:Master formula}
 \{ a_\chi b\} = \sum_{i,j\in I; m,n \geq 0} (-1)^{S(a_{(i,m)},b_{(j,n)})}  b_{(j,n)}\{u_i^{(m)}{ }_{\chi+D} \ u_j^{(n)} \}_{\to}a_{(i,m)}.\end{equation}
\end{remark}

\begin{theorem} \label{Thm:master}
\begin{enumerate}
\item Suppose that $(\mathcal{R}, \{ \, _\chi \, \},D)$ is a SUSY LCA. Then the extension of the $\chi$-bracket to $\mathcal{P}$ via \eqref{Eqn:Master formula2}  endows $\mathcal{P}$ with a SUSY PVA structure.
\item The $\chi$-bracket defined by the super master formula \eqref{Eqn:Master formula2} endows $\mathcal{P}$ with a SUSY PVA structure if and only if the skew-symmetry axiom and the Jacobi identity are satisfied for the elements in the basis $\mathcal{B}$ of $\mathcal{R}$.
\end{enumerate}
\end{theorem}

\begin{proof} (1) We have to show that  \eqref{Eqn:Master formula2} implies the sesquilinearity, skew-symmetry, Jacobi identity and Leibniz rule axioms. Since  \eqref{Eqn:Master formula2} directly implies both the left and the right Leibniz rules, it remains to check sesquilinearity, skew-symmetry and the Jacobi identity.

We begin with the sesquilinearity. 
Denote $a'=Da$ for $a\in \mathcal{P}$. 
Let us check that if $a, b \in \mathcal{P}$ are such that $\{a'_\Lambda c\}= \chi\{a_\Lambda c\}$ and  $\{b'_\chi c\}= \chi\{b_\chi c\}$ for all $c \in \mathcal{P}$, then $\{(ab)'_\chi c\}= \chi\{ab_\chi c\}$ for all $c \in \mathcal{P}$: 
\begin{equation*}
\begin{aligned}
\{(ab)'_\chi c\} & =  \{a'b_\chi c\} +(-1)^{p(a)}\{ab'_\chi c\}  \\
 & =  (-1)^{p(b)p(c)+p(a)+p(c)+1}\{a_{\chi+D} c\}_\to (\chi+D)b  \\
& + (-1)^{(p(a)+1)(p(c)+p(b))}\{b_{\chi+D}c\}_\to a'\\
&+  (-1)^{p(a)+p(b)p(c)+p(c)}\{a_{\chi+D} c\}_\to b'\\
& +(-1)^{ (p(a)+1)(p(b)+p(c))+1} \{b_{\chi+D}c\}_\to (D+\chi) a\\
& = (-1)^{p(b)p(c)}\chi\{a_{\chi+D} c\}_\to b+ (-1)^{ (p(a)(p(b)+p(c))} \chi\{b_{\chi+D}c\}_\to  a
& = \chi\{ab_\chi c\}.
\end{aligned}
\end{equation*}
We conclude that the left sesquilinearity holds by induction, since it holds on the LCA $\mathcal{R}$. Similarly using induction and the right Leibniz rule one shows that the right sesquilinearity holds.  
\\
\indent We can also check the skew-symmetry by induction on the polynomial degree. Indeed, it is clear by sesquilinearity that if 
$\{a _\chi b \}=(-1)^{p(a)p(b)}\{b_{-\chi-D} a\}$ for $a, b \in \mathcal{P}$, then $\{a _\chi b' \}=(-1)^{p(a)(p(b)+1)}\{b'_{-\chi-D} a\}$. Moreover, if $a,b,c \in \mathcal{P}$ are such that $\{a _\chi b \}=(-1)^{p(a)p(b)}\{b_{-\chi-D} a\}$ and $\{a _\chi c \}=(-1)^{p(a)p(c)}\{c_{-\chi-D} a\}$, we conclude from the proof of Lemma \ref{Lem:LLR_2} that $\{a _\chi bc \}=(-1)^{p(a)p(bc)}\{bc_{-\chi-D} a\}$ since left and right Leibniz rules hold.
\\
\indent 
To check the Jacobi identity in Definition \ref{Def:SLCA}, we proceed in three steps. First, we let $a= u_i^{(l)}, b=u_j^{(m)} \in \mathcal{R}$ and $c$ be an arbitrary element in $\mathcal{P}$.  Then 
\begin{equation}\label{Eqn:master_Jacobi_1}
\begin{aligned}
\textstyle \{u_i^{(l)}{}_{\chi} \{ u_j^{(m)}{}_\gamma c\}\} & = \textstyle \sum_{k\in I, n \geq 0} S_1\{u_i^{(l)}{}_\chi c_{(k,n)}\{u_j^{(m)}{}_\gamma u_k^{(n)}\}\}\\
& = \textstyle \sum_{k,p\in I, n,q\geq 0} S_2 \,  c_{(k,n) (p,q)} \{ u_i^{(l)}{}_\chi u_p^{(q)}\} \{u_j^{(m)}{}_\gamma u_k^{(n)}\}\\
&\quad \textstyle  +   \sum_{k\in I, n\geq 0} S_3 \,  c_{(k,n)} \{u_i^{(l)}{}_\chi \{ u_j^{(m)}{}_\gamma u_k^{(n)}\}\},
\end{aligned}
\end{equation}
where $c_{(k,n)(p,q)}=(c_{(k,n)})_{(p,q)}$ and 
\begin{equation}
\begin{aligned}
& S_1= (-1)^{p(c_{(k,n)})(p(u_k^{(n)}u_j^{(m)})+1)}, \\
& S_2=S_1 \, (-1)^{p(c_{(k,n)(p,q)})(p(u_p^{(q)}u_i^{(l)})+1)},\\
& S_3=S_1 \, (-1)^{p(c_{(k,n)})(p(u_i^{(l)})+1)}.
\end{aligned}
\end{equation}
On the other hand, 
\begin{equation}\label{Eqn:master_Jacobi_2}
\{\{u_i^{(l)}{}_\chi u_j^{(m)}\}{}_{\chi+\gamma} c\}= \textstyle \sum_{k\in I, n\geq 0} S_3 \, c_{(k,n)} \{\{ u_i^{(l)}{}_\chi u_j^{(m)}\}{}_{\chi+\gamma} u_k^{(n)}\}
\end{equation}
and 
\begin{equation}\label{Eqn:master_Jacobi_3}
\begin{aligned}
\{u_j^{(m)}{}_\gamma\{u_i^{(l)}{}_\chi  c \}\} &=  \textstyle \sum_{k,p\in I, n,q\geq0}S_4 \,   c_{(k,n)(p,q)} \{u_j^{(m)}{}_\gamma u_p^{(q)}\}\{u_i^{(l)}{}_\chi u_k^{(n)}\} \\
& \textstyle + \sum_{k\in I, n\geq 0}S_3 \, c_{(k,n)} \{u_j^{(m)}{}_\gamma \{u_i^{(l)}{}_\chi u_k^{(n)}\}\},
\end{aligned}
\end{equation}
where $S_4= (-1)^{p(c_{(k,n)})(p(u_i^{(l)}u_k^{(n)})+1)+p(c_{(k,n)(p,q)})(p(u_j^{(m)}u_p^{(q)})+1)}.$ Since $\mathcal{R}$ is a SUSY LCA,  the Jacobi identity holds on the triple $(u_i^{(l)}, u_j^{(m)}, u_k^{(n)})$ 
\begin{equation}\label{Eqn:master_Jacobi_4}
\begin{aligned}
  \{u_i^{(l)}{}_{\chi} \{ u_j^{(m)}{}_\gamma u_k^{(n)}\}\} & = (-1)^{p(u_i^{(l)})+1} \{\{u_i^{(l)}{}_\chi u_j^{(m)}\}{}_{\chi+\gamma} u_k^{(n)}\} \\
  & + (-1)^{(p(u_i^{(l)})+1)(p(u_j^{(m)})+1)} \{u_j^{(m)}{}_\gamma\{u_i^{(l)}{}_\chi u_k^{(n)}\}\}.  
  \end{aligned}
  \end{equation}
By \eqref{Eqn:master_Jacobi_1}, \eqref{Eqn:master_Jacobi_2}, \eqref{Eqn:master_Jacobi_3} and \eqref{Eqn:master_Jacobi_4} we have 
\begin{equation}\label{Eqn:master_Jacobi_5}
\begin{aligned}
 \{u_i^{(l)}{}_{\chi} \{ u_j^{(m)}{}_\gamma c\}\} & = (-1)^{p(u_i^{(l)})+1)} \{\{u_i^{(l)}{}_\chi u_j^{(m)}\}{}_{\chi+\gamma} c\}\\
 & + (-1)^{(p(u_i^{(l)})+1)(p(u_j^{(m)})+1)} \{u_j^{(m)}{}_\gamma\{u_i^{(l)}{}_\chi c\}\}. 
 \end{aligned}
 \end{equation}
 Similarly, using \eqref{Eqn:master_Jacobi_5}, we can show in a second step that the Jacobi identity holds when $a=u_i^{(l)}$ and $b,c$ are arbitrary elements in $\mathcal{P}$. Finally, the same procedure shows that it holds for any elements $a,b,c\in \mathcal{P}$.  
 
 Let us sktech the proof of (2). By (1), it is enough to show that if sesquilinearity, skew-symmetry and the Jacobi identity hold for any elements in $\mathcal{B}$ then they hold for any elements of the form $u_i^{(l)}$. Here, we only show the Jacobi identity. The other properties can be proved by simpler computations. Observe that 
 \begin{equation}\label{auxx}
 \begin{aligned}
 & \{ u_i^{(l)}{}_\chi \{u_j^{(m)}{}_\gamma u_k^{(n)}\}\}= \chi^l (-1)^{m(i+1)}\gamma^m (-1)^{(i+j)n} (\chi+\gamma+D)^n \{ u_i{}_\chi \{u_j{}_\gamma u_k\}\},\\
 & \{\{u_i^{(l)}{}_\Lambda u_j^{(m)}\}{}_{\chi+\gamma} u_k^{(n)}\} \\
 & = (-\chi)^l(-1)^{m(i+1)} (-\chi+\chi+\gamma)^m (-1)^{(i+j)n} (\chi+\gamma+D)^n \{\{u_i{}_\chi u_j\} {}_{\chi+\gamma} u_k\},\\
 & \{u_j^{(m)}{}_\gamma \{ u_i^{(l)}{}_\chi u_k^{(n)}\}\}= \gamma^m (-1)^{l(j+1)} \chi^l (-1)^{(i+j)n} (\chi+\gamma+D)^n \{ u_j{}_\gamma \{u_i{}_\chi u_k\}\}.
 \end{aligned}
 \end{equation}
 By the Jacobi identity on the triple $(u_i,u_j,u_k)$
\[   \{u_i{}_{\chi} \{ u_j{}_\gamma u_k\}\}  = (-1)^{p(u_i)+1} \{\{u_i {}_\chi u_j\}{}_{\chi+\gamma} u_k\}  + (-1)^{(p(u_i)+1)(p(u_j)+1)} \{u_j{}_\gamma\{u_i{}_\chi u_k\}\}, \]
and equation \eqref{auxx}, we have 
\begin{equation*}
\begin{aligned}
 \{ u_i^{(l)}{}_\chi \{u_j^{(m)}{}_\gamma u_k^{(n)}\}\}& = (-1)^{p(u_i^{(l)})+1} \{\{u_i^{(l)}{}_\chi u_j^{(m)}\}{}_{\chi+\gamma} u_k^{(n)}\} \\
 & (-1)^{(p(u_i^{(l)})+1)(p(u_j^{(m)})+1)} \{u_j^{(m)}{}_\gamma \{ u_i^{(l)}{}_\chi u_k^{(n)}\}\}.
\end{aligned}
\end{equation*}
\end{proof}

The following example and its reductions will play a key role in the next sections.
\begin{example}[affine SUSY PVA]  \label{Ex:affine}
Let $\g$ be a simple finite Lie superalgebra with an even nondegenerate invariant bilinear form $(\, |\, ).$ Consider the parity reversed vector superspace $\overline{\g}=\{\bar{a}|a\in \g\}$, that is  $(-1)^{p(\bar{a})}= (-1)^{p(a)+1}$. Then, the differential superalgebra $\mathcal{P}(\overline{\g}):= S(\mathbb{C}[D]\otimes \overline{\g})$ endowed with the $\chi$-bracket 
\begin{equation}\label{Eqn:affine}
 \{\bar{a}_\chi \bar{b}\}_1=(-1)^{p(a)} (\overline{[a,b]}+ \chi(a|b)) \quad \text{ for } a,b\in \g
 \end{equation}
 is a SUSY PVA, called the affine SUSY PVA associated to $\g$. One can easily check that the bracket satisfies skew-symmetry and the Jacobi identity so that it actually induces a SUSY PVA structure on $\mathcal{P}(\overline{\g})$ by Theorem \ref{Thm:master}.

The $\chi$-bracket on $\mathcal{P}(\overline{\g})$ defined by 
\begin{equation}\label{Eqn:affine_s}
\{\bar{a}_\chi \bar{b}\}_2=(-1)^{p(a)+1} (s|[a,b]) \quad \text{ for } a,b\in \g 
\end{equation}
where $s\in \g_{\bar{1}}$, also defines a SUSY PVA structure on $\mathcal{P}(\overline{\g})$, which is compatible with the first structure \eqref{Eqn:affine}, meaning that any of their linear combination induces a SUSY PVA structure.
\end{example}

\begin{example}[SUSY classical W-algebras] \label{W-algebra} 
Let $\g=\bigoplus_{i\in \Z} \g_i$ be a Lie superalgebra as in Example \ref{Ex:affine} with a $\Z$-grading $[\g_i, \g_j]\subset \g_{i+j}$. Let $f\in \g_{-1}\setminus \{0\}$ and  $s\in \g_d$ be odd elements, where $\g_{d'}=0$ for $d'>d$. Consider the vector superspace 
\begin{equation}\label{Eqn:W(g)}
\mathcal{W}({\g},f,s):= (\mathcal{P}(\overline{\g})/ \mathcal{I}_f)^{ad_\chi \overline{\n}}
\end{equation}
where $\overline{\n}:=  \overline{\g}_{>0}$ and $\text{ad}_\chi \bar{n} (A):= \{\bar{n}_\chi A\}$ for the bracket defined by \eqref{Eqn:affine_s} and $\bar{n}\in \overline{\n}$.
Then $\mathcal{W}(\g,f,s)$ is a SUSY PVA whose $\chi$-bracket is induced from \eqref{Eqn:affine_s}. The SUSY PVA $\mathcal{W}(\g,f,s)$ can be understood as a classical limit of the SUSY W-algebra introduced in \cite{MR}.

\end{example}

\subsection{Supersymmetric Hamiltonian systems}
We define super Hamiltonian equations and super integrable systems in the framework of SUSY PVAs. Let $V$ be a vector superspace with a basis $\mathcal{B}$. Suppose that $\mathcal{P}=S(\mathbb{C}[D] \otimes V)$ is endowed with a SUSY PVA structure. The super Hamiltonian equation associated to a Hamiltonian $\smallint h$ can be identified with the derivation $X_{\smallint h}$ introduced in Definition \ref{defXf}. 

\begin{definition} \label{defint}\hfill
\begin{enumerate}
\item We call elements of $\mathcal{P}$ functions and elements of $\mathcal{P} /D \mathcal{P}$ functionals.
\item The super Hamiltonian system of equations associated to the Hamiltonian $\int h \in \mathcal{P}/D \mathcal{P}$ is the system
\begin{equation}\label{SHE}
 \frac{du_i}{dt}= \{ \smallint h{}_\chi u_i\}|_{\chi=0}=X_{\smallint h}(u_i), \, \, i \in I.
\end{equation}
\item
We say that $\smallint \rho$ is an \textit{integral of motion} of the super Hamiltonian equation \eqref{SHE} if $\{ \smallint h, \smallint \rho \}=0$. It is equivalent to  $\smallint X_{\smallint h}( \rho)=0$. 
\item
The super Hamiltonian equation  \eqref{SHE} is called \textit{integrable} if $\smallint h$ is contained in an infinite dimensional abelian subalgebra $\mathfrak{a}$ of $(\mathcal{P}/D \mathcal{P}, \{ \ , \ \})$. In other words, the integrable super Hamiltonian equation possesses infinitely many linearly independent integrals of motion. Moreover, for any $\smallint h_1, \smallint h_2 \in \mathfrak{a}$, we have 
\begin{equation}
[X_{\smallint h_1},X_{\smallint h_2}]=0
\end{equation}
by Lemma \ref{eqcom}. We say that the super Hamiltonian equations associated to $\smallint h_1$ and $\smallint h_2$ are \textit{compatible}.

\end{enumerate}

To conclude Section \ref{Sec: SUSY PVA and sHIS} we prove Corollary \ref{cor_1&2}, which will be used at the end of Section \ref{Sec:DS hier}. We begin by recalling the definition of the \textit{variational derivative}.
\end{definition}
\begin{definition}
The variational derivative of a functional $a\in \mathcal{P}$ with respect to the variable $u_i$ is 
\[  \frac{\delta a}{\delta u_i}:=\sum_{m \geq 0} (-1)^{mi+\frac{m(m+1)}{2}} D^m (\frac{\partial a}{\partial u_i^{(m)}}).\]
If $\int a=0$, it follows from Lemma \ref{Lemma: commutator_derivatives} that $\frac{\delta a}{\delta u_i}=0$.
\end{definition}

\begin{proposition}\hfill
\begin{enumerate}
\item
The equation \eqref{SHE} can be rewritten as
\begin{equation} \label{z}
 \frac{du_i}{dt}=\sum_{j\in I} (-1)^{i(p(h)+j)}\{u_j{}_D u_i\}_{\to}\frac{\delta a}{\delta u_j}, \, \, \, i\in I.
 \end{equation}
\item For $a,b\in \mathcal{P}$, we have 
\begin{equation} \label{Eqn:mater formula-local}
 \smallint \{ a{}_\chi b\}|_{\chi=0} =\sum_{i,j\in I} (-1)^{p(b)p(a)+p(b)(j+1)+ij} \int\frac{\delta b}{\delta u_j} \{ u_i{}_D u_j\}_{\to} \frac{\delta a}{\delta u_i}.
 \end{equation}
\end{enumerate}
\end{proposition}
\begin{proof}
It follows directly from the super master formula \eqref{Eqn:Master formula2}. 
\end{proof}
\begin{remark}
We deduce from equations \eqref{z} and \eqref{Eqn:mater formula-local} that for any $a,b\in \mathcal{P}$, we have
\begin{equation} \label{ref}
\begin{split}
 \textstyle \smallint \{ a{}_\chi b\}|_{\chi=0} &  \textstyle  =\sum_{j\in I} (-1)^{p(b)(p(a)+1)+j(p(a)+p(b))} \int\frac{\delta b}{\delta u_j}  \{ a{}_\chi u_j\}|_{\chi=0} \\
  &  \textstyle  =\sum_{j\in I} (-1)^{j} \int \big( \{ a{}_\chi u_j\}|_{\chi=0} \big) \frac{\delta b}{\delta u_j}.
  \end{split}
\end{equation}
\end{remark}

\begin{corollary} \label{cor_1&2}
Suppose that $\mathcal{P}$ admits two compatible SUSY PVA brackets $\{ \, _\chi \, \}_1$ and $\{ \, _\chi \, \}_2$.
Let $\int h_1, \int h_2  \in \mathcal{P}/D \mathcal{P}$ be such that $\{ \int h_1, \int g\}_1=\{\int h_2, \int g \}_2$ for all $\int g \in \mathcal{P}/D \mathcal{P}$. Then ${\{{ h_1} _\chi g\}_1}|_{\chi=0}={\{{ h_2} _\chi g\}_2}|_{\chi=0}$ for all $g \in \mathcal{P}$.
\end{corollary}
\begin{proof}
The corollary follows from equation \eqref{ref}. Indeed, it suffices to take $g=(u_j^{(2n)})^3$ with large $n$ if $j$ is even and $g=(u_j^{(2n+1)})^3$ if $j$ is odd.
\end{proof}

\section{SUSY Drindeld-Sokolov reductions}  \label{Sec: DSred}

In this section and the next, we let $\g$ be a Lie superalgebra equipped with a supersymmetric invariant nondegenerate bilinear form $(\, |\, )$. We assume that $\g$ is $\Z$-graded with degrees ranging from $-d$ to $d$ for some $d >0$ and that the form $(\, |\, )$ pairs $\g_k$ with $\g_{-k}$ for all $k=0,\cdots,d$. 
We consider a sub Lie superalgebra $\n$ of $\g_{>0}$ which is homogeneous for the $\Z$-grading. Let $\mathfrak{m} \subset \n$ be a homogeneous $\n$-module.  We denote by $\mathfrak{b}_-$ a homogeneous complement to $\mathfrak{m}$ in $\g$ and by $\mathfrak{b}$ the dual space of $\mathfrak{b}_-$ with respect to $(\, |\, )$. Then $\mathfrak{b}$ is the orthogonal of $\mathfrak{m}$. We pick two positive integers  $i,j \leq d$ and odd elements 
\begin{equation*}
f \in \g_{-i},  \, \, s \in \g_j.
\end{equation*}
 Finally, we assume that
 \begin{itemize}
\item[(A-1)] $ \g_{\geq i} \subset \mathfrak{m}$,
\item[(A-2)]$[f, \n] \subset \mathfrak{b}$,
\item[(A-3)] $\text{ad}f|_{\n} \text{ is injective},$
\item[(A-4)] $[s, \n]=0.$
\end{itemize}
It is clear from the invariance of the bilinear form that $\mathfrak{b}$ is a $\n$-module. We fix two dual (for the bilinear form) bases $\{q^t| t \in \tilde{I} \}$ and $\{q_t| t \in \tilde{I}  \}$ of $\g$ 
such that $\tilde I$ is a subset of $\Z$. We assume that the parity of $q^t$ is the same than the parity of $t$. Thus we have 
\[ (q^t|q_{t'})= (-1)^t (q_{t'}|q^t)= \delta_{t,t'} \text{ for } t,t'\in \tilde I.\]
In the sequal, we suppose
\[ \mathfrak{b}= \text{Span}_{\C}\{q^t|t\in I \subset \tilde I\}, \quad \mathfrak{b}_-= \text{Span}_{\C}\{q_t|t\in I \subset \tilde I\}.\]

\begin{example} \label{Example: f,m,n_W}
Let $\g$ be a simple Lie superalgebra with a subalgebra $\mathfrak{h}$ isomorphic to $\mathfrak{osp}(1|2)$.  $\mathfrak{h}$ is generated by two odd elements $e$ and $f$ and a $\text{sl}_2$-triple $(E,H,F)$.  The element $H$ induces a $\Z$-grading on $\g$ with homogeneous subspaces $\g_i:=\{ a\in \g| [h,a]= i a\}$ and $f$ is in $\g_{-1}$. We can pick $\mathfrak{n}=\mathfrak{m}=\g_{>0}$ so that $\mathfrak{b}= \g_{\geq 0}$, $\text{ad}f|_{\mathfrak{n}}$ is injective and $[f, \mathfrak{n}]\subset \mathfrak{b}$. This data corresponds to the SUSY PVA in Example \ref{W-algebra}. 
However, in this case, we cannot guarantee the existence of an odd element  $s$ satisfying (A-4). This has to be checked case by case.
\end{example}

We will give examples of tuples $(\g, \mathfrak{m}, \n, f,s)$ satisfying (A-1), (A-2), (A-3), (A-4) in Section \ref{Subsec:Example-(g,f,m,n,s)}.

\subsection{Generalized SUSY W-algebras}
We recall that the differential superalgebra $\mathcal{P}(\overline \g)$ is equipped with two compatible SUSY PVA $\chi$-brackets defined by
 \begin{equation}\label{affinechibrackets}
 \begin{split}
 \{ \bar a _\chi \bar b\}_1&=(-1)^{p(a)}(\overline{ [a,b]}+ \chi (a|b)), \\
  \{ \bar a _\chi \bar b\}_2&=(-1)^{p(a)+1}(s|[a,b]), \, \, \, \, \, a, b \in \g.
\end{split}
 \end{equation}
 
 In the following two propositions, we introduce the generalized SUSY $W$-algebra asssociated to the tuple $(\g, \n, \mathfrak{m}, f)$ and show that it admits two compatible SUSY PVA $\chi$-brackets.

\begin{proposition} \label{Prop: new W-algebra}
Let $\mathcal{I}_{f, \mathfrak{m}}$ be the differential  ideal of the differential superalgebra $\mathcal{P}(\overline{\g})$ generated by $\{ \, \bar{m}-(f|m)\, | \, m\in \mathfrak{m}\}$. Let $\pi$ denote the projection $\mathcal{P}(\overline{\g}) \rightarrow \mathcal{P}(\overline{\g})/\mathcal{I}_{f, \mathfrak{m}}$. Then the following differential superalgebra is well-defined:
\begin{equation} \label{generalized W-agebra}
 \mathcal{W}(\g, \n, \mathfrak{m}, f):=\{ \pi(a) \in  \mathcal{P}(\overline{\g})/\mathcal{I}_{f, \mathfrak{m}} | \{ \n _\chi a \}_1 \subset \mathcal{I}_{f, \mathfrak{m}}[\chi] \}.
 \end{equation}
Moreover, the following bracket defines a SUSY PVA structure on 
 $\mathcal{W}(\g, \n, \mathfrak{m}, f)$
\begin{equation} \label{bracket_H}
 \{ \pi(a) _\chi \pi(b) \}_1^{\mathcal{W}}:=\pi (\{a _\chi b \}_1) \, \text{ for }  \pi(a) ,\pi(b) \in \mathcal{W}(\g, \n, \mathfrak{m}, f).
\end{equation}
\end{proposition}

\begin{proof}
 Since $\mathfrak{m}$ is a $\n$-module and $[f, \n] \subset \mathfrak{b}$,  we have for all $n \in \mathfrak{n}$ and $ m \in \mathfrak{m}$
 \begin{equation*}
  \{ \bar{n}{}_\chi \bar{m}\}_1=(-1)^{p(n)}\overline{[n,m]}=(-1)^{p(n)}(\overline{[n,m]}-(f |[n,m])).
  \end{equation*}
  In other words, $ \{ \bar{n}{}_\chi \bar{m}\}_1 \in \mathcal{I}_{f, \mathfrak{m}}$ and by sesquilinearity, $ \{ \bar{n}{}_\chi \bar{m}^{(k)}\}_1 \in \mathcal{I}_{f, \mathfrak{m}}[\chi]$  for all $k \geq 0$. Let $a  \in \mathcal{P}(\overline{\g})$. For all $n \in \n$ and$ \in \mathfrak{m}$, we have
\begin{equation*}
 \{ \bar{n}{}_\chi a(\bar{m}-(f|m))^{(k)} \}_1 = \{\bar{n}{}_\chi a\}_1 (\bar{m}-(f|m))^{(k)}  +(-1)^{p(n)p(a)} a 
\{ \bar{n}{}_\chi \bar{m}^{(k)}\}_1\in \mathcal{I}_{f, \mathfrak{m}}[\chi].
 \end{equation*}
by the Leibniz rule. 
 Thus we proved that $\{ \overline{ \mathfrak{n}} _\chi \mathcal{I}_{f, \mathfrak{m}} \}_1 \subset  \mathcal{I}_{f, \mathfrak{m}}[\chi]$.
Hence  $\mathcal{W}(\g, \n, \mathfrak{m}, f)$ is well-defined as a vector superspace. It clearly is a sub-superalgebra of  $\mathcal{P}(\overline{\g})/\mathcal{I}_{f, \mathfrak{m}} \simeq S(\C[D] \otimes \overline{\mathfrak{b}}_{-})$, by the sesquilinearity and the Leibniz rule.

 Let $a  \in \mathcal{P}(\overline{\g})$ be such that $\{ \overline{\n} {}_\chi a \}_1 \subset \mathcal{I}_{f, \mathfrak{m}}[\chi]$. In particular $\{ \overline{\mathfrak{m}} {}_\chi a \}_1 \subset \mathcal{I}_{f, \mathfrak{m}}[\chi]$ since $\mathfrak{m} \subset \n$. It follows, using  sesquilinearity and the Leibniz rule, that $\{ a {}_\chi \mathcal{I}_{f, \mathfrak{m}} \}_1 \subset \mathcal{I}_{f, \mathfrak{m}}[\chi]$.
Therefore the induced $\chi$-bracket \eqref{bracket_H} on $\mathcal{W}(\g, \n, \mathfrak{m}, f)$ is well defined. It is straightforward to check that  axioms of a SUSY PVA hold for this $\chi$-bracket using the even morphism of differential superalgebras $\pi$.
\end{proof}

\begin{proposition}\label{keyprop}
The following $\chi$-bracket is a SUSY PVA bracket on $\mathcal{W}(\g, \n, \mathfrak{m}, f)$ compatible with the first bracket \eqref{bracket_H}
\begin{equation}\label{secondbra}
 \{ \pi(a){}_\chi \pi(b)\}_2^{\mathcal{W}}:= \pi(\{a {}_\chi b\}_2) \text{ for }  \pi(a) ,\pi(b) \in \mathcal{W}(\g, \n, \mathfrak{m}, f).
\end{equation}
\end{proposition}

\begin{proof}
We know from Example \ref{Ex:affine} that for any constant $\epsilon$ the bracket 
\begin{equation} \label{bracketeps}
\{a_\chi b\}^{\epsilon}=\{a_\chi b\}_1+\epsilon \{a_\chi b\}_2, \, \, a,b \in \mathcal{P}(\overline{\mathfrak{g}})
\end{equation}
defines a SUSY PVA structure on $\mathcal{P}(\overline{\mathfrak{g}})$. 

Let us show $\mathcal{W}(\g, \n, \mathfrak{m}, f)$ is SUSY PVA with the $\chi$-bracket induced from \eqref{bracketeps} which directly implies the proposition.
Since the bilinear form is invariant and $s \in \text{ker} \, \text{ad} \, \mathfrak{n}$, we have $\{\bar{n}{}_\chi a\}_2=0$ for any $a\in \overline{\g}$ and $n \in \n$.  By the Leibniz rule, $\{\bar{n}{}_\chi u\}_2=0$ for any $u\in \mathcal{P}( \overline{\g})$ and $n \in \n$. It follows from this observation that, for an element $u \in  \mathcal{P}(\overline{\mathfrak{g}})$, $\{\overline{\mathfrak{n}} _\chi u\}_1\subset \mathcal{I}_{f, \mathfrak{m}}[\chi]$ if and only if $\{\overline{\mathfrak{n}} _\chi u\}^{\epsilon} \subset \mathcal{I}_{f, \mathfrak{m}}[\chi]$. Therefore the exact same proof as in Proposition \ref{Prop: new W-algebra} gives the result. In the definition \eqref{generalized W-agebra} of $\mathcal{W}(\g, \n, \mathfrak{m}, f)$, one can replace the first $\chi$-bracket by the $\chi$-bracket \eqref{bracketeps}.
\end{proof}

\begin{definition}
The differential superalgebra $\mathcal{W}(\g, \n, \mathfrak{m}, f)$ introduced in Proposition \ref{Prop: new W-algebra} is called the \textit{generalized W-algebra} associated to $\g$, $\n$, $\mathfrak{m}$ and $f$.
\end{definition}

\begin{remark}
The SUSY classical W-algebra $\mathcal{W}(\g,f,s)$ in Example \ref{W-algebra} is $\mathcal{W}(\g,\g_{>0}, \g_{>0},f)$.
\end{remark}

\subsection{Generalized W-algebras and Lax operators} 
The goal of this section is to give an equivalent definition of $\mathcal{W}(\g, \n, \mathfrak{m}, f)$ in terms of super differential operators on $\g \otimes \mathcal{P}(\overline \g)$ (Theorem \ref{Lem:equi_W_1}).
Consider the following sub differential superalgebra of $\mathcal{P}(\overline{\mathfrak{g}})$
\begin{equation}
\textstyle \mathcal{P}:=  S(\C[D]\otimes \overline{\mathfrak{b}}_-).
\end{equation}
Note that $\mathcal{P}(\overline \g)$ splits as the direct sum of $\mathcal{P}$ and $\mathcal{I}_{f, \mathfrak{m}}$.
We endow
$\g \otimes \mathcal{P}$
  with the Lie superalgebra bracket
\begin{equation} \label{LieBracket_part1}
[a\otimes u , b\otimes v]= (-1)^{p(b) p(u)} [a,b] \otimes uv, \, \, a,b \in \g, \, u,v, \in \mathcal{P}
\end{equation}
and with the $\mathcal{P}$-valued bilinear form
\begin{equation}\label{Bilinear form}
(a\otimes u|b\otimes v)= (-1)^{p(b)p(u)} (a|b) uv,  \, \, a,b \in \g, \, u,v, \in \mathcal{P}.
\end{equation}
We extend this Lie superalgebra bracket to the semi-direct product $\mathbb{C}[D] \ltimes (\g \otimes \mathcal{P})$ 
\begin{equation*}
 [D, a\otimes u]= (-1)^{p(a)} a\otimes u',  \, \, a \in \g, \, u \in \mathcal{P}.
 \end{equation*}

 \begin{definition} Let $\mathcal{F}:= (\mathfrak{b}\otimes \mathcal{P})_{\bar{1}}$. It is a $(\mathfrak{n} \otimes \mathcal{P})_{\bar{0}}$-submodule of $\g \otimes \mathcal{P}$. 
 \begin{enumerate}
\item  For each $\mathcal{Q} \in \mathcal{F}$, we call the odd differential operator 
\begin{equation}
 \mathcal{L}(\mathcal{Q})= D + \mathcal{Q} +f \otimes 1 \in \mathbb{C}[D] \ltimes (\g \otimes \mathcal{P})
 \end{equation}
the {\it Lax operator} associated to $\mathcal{Q}$. 
\item When $\mathcal{Q}_{u}=  \sum_{t\in I} {q^t} \otimes \bar{q}_t \in \mathcal{F},$ we call the corresponding Lax operator
 \begin{equation}
 \mathcal{L}_u=D + \sum_{t\in I} {q^t} \otimes \bar{q}_t  +f \otimes 1
 \end{equation}
 the {\it universal Lax operator} associated to $f$ and $\mathfrak{b}$. Note that $\mathcal{Q}_u$ and $\mathcal{L}_u$ are independent on the choice of basis  $\{q^t\}_{t\in I}$ and $\{q_t\}_{t\in I}$, provided that they are dual to each other for the bilinear form. 
\item Two elements $\mathcal{Q}_1$ and $\mathcal{Q}_2$ in $\mathcal{F}$ are called {\it gauge equivalent} if there exists $U\in (\mathfrak{n}\otimes  \mathcal{P})_{\bar{0}}$  such that 
 \[ e^{\text{ad}U}  \mathcal{L}(\mathcal{Q}_1)=  \mathcal{L}(\mathcal{Q}_2). \]
 \end{enumerate}
\end{definition}

We assume that $\{q^t|t \in J \subset I \}$ is a basis of $[f, \n]$ and we let $V:= \text{Span}_\C\{ q^t| t\in I\setminus J \}$
 so that we have the direct sum decomposition \[ V\oplus [\n, f]=  \mathfrak{b}. \] 
Following the lines of \cite{IK92} and \cite{IK91}, we perform a supersymmetric Drinfeld Sokolov reduction on our universal Lax operators.
\begin{lemma} \label{Lem:Lax_can}
For all $\mathcal{Q}$ in $\mathcal{F}$, there exists a unique $\mathcal{Q}^c \in (V \otimes \mathcal{P})_{\bar{1}} \subset \mathcal{F}$ such that $\mathcal{Q}$ and $\mathcal{Q}^c$ are gauge equivalent. $\mathcal{L}(\mathcal{Q}^c)$ is called the \emph{canonical form} of $\mathcal{L}(\mathcal{Q})$. 
\end{lemma}

\begin{proof}
We fix $\mathcal{Q}$ and simply denote $\mathcal{L}(\mathcal{Q})$ by $L$.
Let us find $N\in (\mathfrak{n} \otimes \mathcal{P})_{\bar{0}}$ such that $L^c:=e^{\text{ad} \, N} L$ is of the desired form. Denote by $L^c_{k}$, $N_k$ and $L_k$ the $(\g_k \otimes \mathcal{P})$-parts of $L^{c}$, $N$ and $L$.  We must have 
\[ L_k= L^c_{k}=0\,  \text{  for  } \, k<-i. \]
Since $L_{-i}=f\otimes 1$, the $\g_{-i+1} \otimes \mathcal{P}$-part of the equation $L^c:=e^{\text{ad} \, N} L$ is
\begin{equation} \label{step_1}
 L^c_{-i+1}=L_{-i+1}+ [N_{1}, f\otimes 1].
\end{equation}
There is a unique decomposition $L_{-i+1}=L_{-i+1}^V+L_{-i+1}^{\perp}$ such that   $L_{-i+1}^V\in V\otimes \mathcal{P}$ and  $L_{-i+1}^{\perp}\in [\mathfrak{n}, f]\otimes \mathcal{P}$. Since  $\text{ad} f|_{\mathfrak{n}}$ is injective, there exists a unique $N_{1}\in (\n \otimes \mathcal{P})_{\bar{0}}$ such that $L_{-i+1}^{\perp}=[N_{1}, f\otimes 1]$. Moreover, $L^c_{-i+1}=L_{-i+1}^{\perp}$ in \eqref{step_1} is also unique.

Let us fix $t$ such that $-i+1\leq t\leq d$. For any $t'<t$, suppose that $L^c_{t'}$ and $N_{t'+i}$  are uniquely defined .
Then the $(\g_t \otimes \mathcal{P})$-part of the equation $L^c:=e^{\text{ad} \, N}L$ can be written as
\[L^c_t=  (e^{\text{ad} N_{<t+i}} L)_t + [N_{t+i}, f\otimes 1],\]
where $N_{<t+i}= \sum_{t'<t} N_{t'+i}.$ By the induction hypothesis, $(e^{\text{ad} N_{<t+i}} L)_t$ is given, hence  $N_{t+i}$ and $L^c_t$ exist and are uniquely determined by the injectivity of  $\text{ad} f|_{\mathfrak{n}}$.
\end{proof}

\begin{definition}
We identify an element in $\mathcal{P}$ with a map from $\mathcal{F}$ to $\mathcal{P}$ as follows. Let $a \in \mathcal{P}$.  If 
$\mathcal{Q}=\sum_{t \in I}{ q^{t} \otimes  p_t}$, we denote by $a(\mathcal{Q})$ the element of $\mathcal{P}$ obtained from $a$ by substituting ${\bar{q}_t}^{(n)}$ with ${p_t}^{(n)}$ for all $n \geq 0$ and all $t \in I$.
Under this identification, $a(\mathcal{Q}_u)=a$ for all $a \in \mathcal{P}$. We extend these functions to differential operators by letting $a(\mathcal{L}(\mathcal{Q}))=a(\mathcal{Q})$ for all $\mathcal{Q} \in \mathcal{F}$.
An element $a \in \mathcal{P}$ is called {\it gauge invariant} if $a=a(\mathcal{Q})$ whenever $\mathcal{Q} \in \mathcal{F}$ is gauge equivalent to $\mathcal{Q}_u$.
\end{definition}
\begin{remark}
It follows that for a gauge invariant element $a \in \mathcal{P}$, $a(\mathcal{Q}_1)=a(\mathcal{Q}_2)$ whenever $Q_1$ and $Q_2$ are gauge equivalent in $\mathcal{F}$.
\end{remark}
\begin{definition}
 We denote the canonical form of the universal Lax operator by 
\[ \mathcal{L}^c=  D +\sum_{t\in I\setminus J} q^t \otimes w_t +f\otimes 1.\]
\end{definition}

\

\begin{lemma} \label{Lem:generator}
 The superalgebra of gauge invariant functions in $\mathcal{P}$ is  
 \[S(\C[D]\otimes \text{Span}_{\C}\{ w_t | t\in I\setminus J\}).\]
\end{lemma}

\begin{proof}
For any $\mathcal{Q} \in \mathcal{F}$ we have $\mathcal{L}(\mathcal{Q}^c)= D+ \sum_{t\in I\setminus J} q^t \otimes w_t(\mathcal{Q})+ f\otimes 1$.
 If $\mathcal{Q}$ is gauge equivalent to $\mathcal{Q}_u$, then $\mathcal{Q}^c$ is gauge equivalent to both $\mathcal{Q}$ and $\mathcal{Q}_u$.  By the uniqueness in Lemma \ref{Lem:Lax_can},  $\mathcal{Q}^c=\mathcal{Q}_u^c$ hence $w_t (\mathcal{Q})= w_t.$ Therefore $S(\C[D]\otimes \text{Span}_{\C}\{ w_i | i\in I\setminus J\})$ is included in the set of gauge invariant functions in $\mathcal{P}$. 

Conversely, 
if an element $a \in\mathcal{P}$ is gauge invariant  then 
\begin{equation*}
 a=a(\mathcal{Q}_u)=a(\mathcal{Q}_u^c) \in S(\C[D]\otimes \text{Span}_{\C}\{ w_i | i\in I\setminus J\}).
 \end{equation*}
\end{proof}

The following Lemma essentially says that a function $a \in \mathcal{P}$ is gauge invariant if and only of it is infinitesimally gauge invariant.
\begin{lemma} \label{Lem:gauge_equivalent}
Let $a\in \mathcal{P}$. Then $a \in \mathcal{P}$ is gauge invariant if and only if 
\[ \textstyle \frac{d}{d\epsilon}a(e^{\text{ad}\, n \otimes \epsilon R} (\mathcal{L}_u))|_{\epsilon=0}=0\]
for any $n\otimes R\in (\mathfrak{n} \otimes \mathcal{P})_{\bar{0}}$.
\end{lemma}

\begin{proof}
By definition, if $a \in \mathcal{P}$ is gauge invariant then 
$ \textstyle \frac{d}{d\epsilon}a(e^{\text{ad}\,  n \otimes \epsilon R} (\mathcal{L}_u))|_{\epsilon=0}=0.$ Let us show the converse.
 It is clear that  
 \begin{equation*}
  \frac{d}{d\epsilon}a(e^{\text{ad} n \otimes \epsilon R} (\mathcal{L}_u))|_{\epsilon=0}=0 \implies  \frac{d}{d\epsilon}a(e^{\text{ad} n \otimes \epsilon R} (\mathcal{L}))|_{\epsilon=0}=0
  \end{equation*}
   for any Lax operator $\mathcal{L}(\mathcal{Q})$ since $ \frac{d}{d\epsilon}a(e^{\text{ad} n \otimes \epsilon R} (\mathcal{L}_u))|_{\epsilon=0}(\mathcal{Q})= \frac{d}{d\epsilon}a(e^{\text{ad} n \otimes \epsilon R} (\mathcal{L}(\mathcal{Q})))|_{\epsilon=0}$, where we treated $R$ as an indeterminate.  In particular we have \[\textstyle \frac{d}{d\epsilon_1}a(e^{\text{ad} n \otimes \epsilon_1 R} (e^{\text{ad} n \otimes \epsilon_2 R} \mathcal{L}_u))|_{\epsilon_1=0}= \frac{d}{d\epsilon_1}a(e^{\text{ad} n \otimes \epsilon_1 R} (\mathcal{L}_u))|_{\epsilon_1=\epsilon_2}=0\] for any constant $\epsilon_2$. Therefore  $ \frac{d}{d\epsilon}a(e^{\text{ad} n \otimes \epsilon R}  \mathcal{L}_u)=0$ implying that $a(e^{\text{ad} n \otimes \epsilon R}. \mathcal{L}_u)=a$ for all $\epsilon$, 
 Hence $a$ is gauge invariant.
\end{proof}

\begin{lemma} \label{Lem:equiv}
Recall the definition of the ideal $\mathcal{I}_{f, \mathfrak{m}}$ in Proposition \ref{Prop: new W-algebra}.
For  $a \in \mathcal{P}$ and $n\otimes R \in (\n \otimes \mathcal{P})_{\bar{0}}$, let $X_{2n_0+n_1}\in \mathcal{P}, n_0 \geq 0, n_1\in \{0,1\}$ be defined by the relation 
\begin{equation} \label{eqn:equiv_1}
 \textstyle  \frac{d}{d\epsilon} a(e^{\text{ad} n \otimes \epsilon R} (\mathcal{L}_u))|_{\epsilon=0}= \sum_{n_0 \geq 0, n_1\in \{0,1\}}(-1)^{n_0} X_{2n_0+n_1}  R^{(2 n_0+ n_1)}.
\end{equation}
Then
\begin{equation} \label{eqn:equiv_2}
\textstyle  \{\overline{n} {}_\chi a\}_1 + (-1)^{p(n) p(X)} \sum_{n_0 \geq 0, n_1\in \{0,1\}} X_{2n_0+n_1} \chi^{2 n_0+ n_1} \in  \mathcal{I}_{f, \mathfrak{m}}[\chi],
\end{equation}
where $\{ \ {}_\chi\ \}_1$ is the first $\chi$-bracket \eqref{Eqn:affine} of the affine SUSY PVA $\mathcal{P}(\overline{\g})$.
\end{lemma}

\begin{proof}
Let us show the lemma for $\bar a \in \overline{\mathfrak{b}}_-$. Observe that
\begin{equation} \label{Eqn:gauge_invariant_first}
\begin{aligned}
&  \textstyle \bar{a}([n\otimes R, D+ \sum_{t\in I} q^t \otimes \bar{q}_t + f\otimes 1]) \\
& \textstyle =  \bar{a}((-1)^{p(n)+1}n\otimes R'+ \sum_{t\in I} (-1)^{p(n)} [n, q^t]\otimes \bar{q}_t R + (-1)^{p(n)}[n,f] \otimes R )\\ 
&  \textstyle = (-1)^{p(n)+1} (n|a) R' +  \sum_{t\in I}(-1)^{p(n)} ([n, q^t]|a) \bar{q}_t R + (-1)^{p(n)} ([n,f]|a) R \\
& \textstyle = (-1)^{p(n)p(a)+1} ( \, (n|a) D+ \overline{\rho([n,a])} +(f|[n,a]) \, ) R, 
 \end{aligned}
 \end{equation}
 where $\rho: \g \to \mathfrak{b}_-$ is the orthogonal projection map. Recall that  
 \begin{equation}
 \{ \bar{n}{}_\chi \bar{a}\}_1 = (-1)^{p(n)}( \overline{[n,a]} + \chi (n|a)).
 \end{equation}
Hence the lemma holds for $\bar{a}$ with $X_{0}= (-1)^{p(n)p(a)+1} ( \, \overline{\rho([n,a])} +(f|[n,a]) \, )$ and $X_{1}=(-1)^{p(n)p(a)+1}  (n|a) $.
 
 Suppose  that the lemma is true for $a,b \in \mathcal{P}$, with corresponding coefficients $X_n$ and $Y_n$.  Then we have
  \begin{equation}
\begin{aligned}
 \textstyle  \frac{d}{d\epsilon} ab(e^{\text{ad} n \otimes \epsilon R} (\mathcal{L}_u))|_{\epsilon=0} & \textstyle  = a \sum_{n_0 \geq 0, n_1\in \{0,1\}}(-1)^{n_0} Y_{2n_0+n_1}  R^{(2n_0+n_1)} \\
& \textstyle  +(-1)^{p(a)p(b)}  b \sum_{n_0 \geq 0, n_1\in \{0,1\}} (-1)^{n_0}X_{2n_0+n_1}  R^{(2n_0+ n_1)}.\\
\end{aligned}
\end{equation}
On the other hand, by the Leibniz rule we have 
\begin{equation*}
\{ \bar{n}{}_\chi ab\}_1 = (-1)^{p(n)p(a)} a\{\bar{n}{}_\chi b\}_1 + (-1)^{ (p(n)+ p(a)) p(b)} b\{\bar{n}{}_\chi a\}_1.
\end{equation*} 
Hence
\begin{equation}
\begin{aligned}
 \{\bar{n}{}_\chi ab\}_1& + \textstyle (-1)^{p(n)p(a)} (-1)^{p(n)p(b)} a \sum_{n_0 \geq 0, n_1\in \{0,1\}}  Y_{2n_0+n_1}  \chi^{2n_0+n_1}\\
 & \textstyle +(-1)^{(p(n)+p(a))p(b)} (-1)^{p(n)p(a)} b \sum_{n_0 \geq 0, n_1\in \{0,1\}} X_{2n_0+n_1}  \chi^{2n_0+n_1}\\
 & \in  \mathcal{I}_{f, \mathfrak{m}}[\chi], 
 \end{aligned}
 \end{equation}
 so that the lemma holds for $ab$. Similarly, let us assume that the lemma holds for $a \in \mathcal{P}$ with corresponding coefficients $X_n$. Then 
 \begin{equation}
 \begin{aligned}
 &  \textstyle  \frac{d}{d\epsilon} a'(e^{\text{ad} n \otimes \epsilon R} (\mathcal{L}_u))|_{\epsilon=0}\\   & \textstyle  =  \textstyle D (\sum_{n_0 \geq 0, n_1=0,1} (-1)^{n_0}X_{2n_0+n_1} R^{(2n_0+n_1)})\\
  & = \textstyle  \sum_{n_0 \geq 0, \, n_1\in \{0,1\}} (-1)^{n_0} X_{2n_0+n_1}' R^{(2n_0+n_1)} \\
  &\textstyle  +(-1)^{p(a)+p(n)} \sum_{n_0 \geq 0,n_1\in \{0,1\}}\left(-(-1)^{n_0} X_{2n_0+1} R^{2(n_0+1)}+ (-1)^{n_0}X_{2n_0} R^{(2n_0+1)}\right).
 \end{aligned}
 \end{equation}
  Here we used $p(X_{2n_0+ n_1})= p(a)+p(n)+n_1$.
Since $ \{\bar{n}{}_\chi a'\}_1 = (-1)^n (D+\chi) \{\bar{n}{}_\chi a\}_1$, 
 \begin{equation} \label{DA}
  \{\bar{n}{}_\chi a'\}_1 + (-1)^{p(n)(p(a)+1)}\textstyle  (\chi+D)\sum_{n_0 \geq 0, n_1\in \{0,1\}} X_{2n_0+ n_1}  \chi^{2n_0+n_1}\in \mathcal{I}_{f, \mathfrak{m}}[\chi].
  \end{equation}
The second term in \eqref{DA} can be rewritten as
   \begin{equation}
 \begin{aligned}
& \textstyle  (\chi+D)\sum_{n_0 \geq 0, n_1\in \{0,1\}} X_{2n_0+n_1} \chi^{2n_0+n_1}\\
 & \textstyle = \textstyle  \sum_{n_0 \geq 0, n_1\in \{0,1\}} X'_{2n_0+n_1}\chi^{2n_0+n_1}\\
 & \qquad +\textstyle  \sum_{n_0 \geq 0, n_1\in \{0,1\}}\left( (-1)^{p(a)+p(n)} X_{2n_0+1} \chi^{2n_0+2} + (-1)^{p(a)+p(n)} X_{2n_0} \chi^{2n_0+1} \right).  
 \end{aligned}
 \end{equation}
  Hence the lemma holds for any elements in $\mathcal{P}$.
\end{proof}

Since $\mathcal{P}(\overline \g)=\mathcal{P} \oplus \mathcal{I}_{f, \mathfrak{m}}$, we regard $\mathcal{W}(\g, \n, \mathfrak{m}, f)$ as a sub- differential superalgebra of $\mathcal{P}$ via the isomorphism
\begin{equation} \label{identify}
 \iota: \mathcal{P}(\overline{\g})/\mathcal{I}_{f, \mathfrak{m}} \to \mathcal{P}, \quad \pi({a}+{b}) \mapsto a, \, \, a \in \mathcal{P}, \, b \in \mathcal{I}_{f, \mathfrak{m}}.
 \end{equation}

\begin{theorem}  \label{Lem:equi_W_1} \ 
\begin{enumerate}
\item  The generalized W-algebra $\mathcal{W}(\g, \n, \mathfrak{m}, f)$ is the superalgebra of  gauge invariant functions in $\mathcal{P}$.
 \item $\mathcal{W}(\g, \n, \mathfrak{m}, f)$ is generated by $\{ w_t | t\in I\setminus J\}$ as a differential superalgebra, where $w_t$'s are in Lemma \ref{Lem:generator}.
 \end{enumerate}
\end{theorem}

\begin{proof}
We know that $\mathcal{P} \cap \mathcal{I}_{f, \mathfrak{m}}=0$. Hence,  there is no distinct sequence $(Y_{N})_{N \geq 0}$ in $\mathcal{P}$ which can replace $(X_{N})_{N \geq 0}$  in   \eqref{eqn:equiv_2}. Therefore by  Lemma \ref{Lem:equiv},
\[ \{\overline{n}{}_\chi a\}_1\in  \mathcal{I}_{f, \mathfrak{m}}[\chi] \text{ for all} \, \,  n \in \mathfrak{n} \iff  \textstyle \frac{d}{d\epsilon}a(e^{\text{ad}\, n \otimes \epsilon R} (\mathcal{L}_u))|_{\epsilon=0}=0 \, \text{ for all} \, \, n \otimes R \in (\mathfrak{n} \otimes \mathcal{P})_{\bar{0}}. \]
Using Lemma \ref{Lem:gauge_equivalent}, we can prove (1). (2) follows from Lemma \ref{Lem:generator}.
\end{proof}

\subsection{SUSY PVA structures and Lax operators} 
We express in terms of Lax operators the Lie superalgebra brackets induced from the two compatible affine $\chi$-brackets and from the two  compatible $\chi$-brackets on the generalized W-algebra $\mathcal{W}(\g, \n, \mathfrak{m},f)$.
\begin{lemma}\label{Lemma:Lax operator-affine}
Let $a,b$ be two elements in the affine SUSY PVA $\mathcal{P}(\overline{\g})$. Then
\begin{equation} \label{universal_bracket}
\begin{aligned}
&  \textstyle \int\{ a_\chi b\}_1|_{\chi=0}=- \int (\sum_{t\in \tilde{I}} q_j \otimes \frac{\delta a}{\delta \bar{q}_t}| [ D+ \sum_{k\in \tilde{I}} q^k\otimes \bar{q}_k, \sum_{l\in \tilde{I}} q_{l} \otimes \frac{\delta{b}}{\delta{\bar{q}_{l}}}]), \\
&  \textstyle \int\{ a_\chi b\}_2|_{\chi=0}= \int (\sum_{t\in \tilde{I}} q_t \otimes \frac{\delta a}{\delta \bar{q}_t}| [s\otimes 1, \sum_{k\in \tilde{I}} q_{k} \otimes \frac{\delta{b}}{\delta{\bar{q}_{k}}}]).
\end{aligned}
\end{equation}
\end{lemma}
\begin{proof}
\if 
Take $f=\bar{q}_\alpha$ and $g=\bar{q}_\beta$  for $\alpha, \beta \in \tilde{I}$. Then 
\begin{equation}
\begin{aligned}
 &  \textstyle - (q_\alpha \otimes 1| [ D+ \sum_{i\in \tilde{I}} q^i\otimes \overline{q}_i, q_\beta \otimes 1])  \\
 & \textstyle= (-1)^{\beta i +\beta+1}(q_\alpha|[q^i, q_\beta])\overline{q}_i \textstyle = (-1)^{\beta i +\beta+1+\alpha i +1}(q^i|[q_\alpha, q_\beta])\overline{q}_i \\
 & =(-1)^\alpha \overline{[q_\alpha, q_\beta]} = \{\overline{q}_\alpha{}_\chi \overline{q}_\beta\}_1|_{\chi=0}.
\end{aligned}
\end{equation}
Here we used $(-1)^i= (-1)^{\alpha+\beta}$ if $(q^i|[q^\alpha, q^\beta]) \neq 0$. {\red why do you need the first part of the proof}
\\
\fi
By a direct computation, we have for all $a,b \in \mathcal{P}(\overline \g)$
\begin{equation}\label{aaaux}
\begin{aligned}
& \textstyle - \int  (\sum_{t\in \tilde{I}} q_t \otimes \frac{\delta a}{\delta \bar{q}_t}| [ D+ \sum_{k\in \tilde{I}} q^k\otimes \bar{q}_k, \sum_{l\in \tilde{I}} q_{l} \otimes \frac{\delta{b}}{\delta{\bar{q}_{l}}}]) \\
& = \textstyle \int \sum_{t,k\in \tilde{I}} (-1)^{p(a)p(b)+p(b)k+(t+1)(k+1)}\frac{\delta{b}}{\delta{\bar{q}_{k}}} (-1)^{t}(\overline{[q_t, q_{k}]}+ (q_t|q_{k})D)\frac{\delta a}{\delta \bar{q}_t}.
\end{aligned}
\end{equation}
By \eqref{Eqn:mater formula-local} and since $p(\bar{q}_t) = t+1$, the RHS of \eqref{aaaux} is the same as $\{a{}_\chi b\}_1|_{\chi=0}$ and the first equality of \eqref{universal_bracket} holds. The second line of \eqref{universal_bracket} can be proved analogously.
\end{proof}
\if
We are now going to reformulate the two compatible super Lie brackets on \\ $\mathcal{W}(\g, f, \mathfrak{b},\n)/ D \mathcal{W}(\overline{\g}, f,\mathfrak{b}, \n)$ in terms of the universal Lax operator $\mathcal{L}_u$. Recall that we regard $\mathcal{W}(\g, f, \mathfrak{b},\n)$ as a subalgebra of $\mathcal{P}$ via the map $\iota$ in \eqref{identify}. 
For example, we simply write
\[\textstyle  \frac{\delta \phi}{\delta q_t} := \frac{\delta (\iota(\phi))}{\delta q_t}\]
for $\phi\in \mathcal{W}(\overline{\g}, f, \mathfrak{b}, \n)$ and $t\in I$. 
\fi

\begin{proposition}\label{pro}
The Lie superalgebra brackets of degree $\bar{1}$ on $\mathcal{W}(\g, \n, \mathfrak{m}, f)/ D \mathcal{W}(\g, \n, \mathfrak{m}, f)$ induced from   \eqref{bracket_H} and \eqref{secondbra} are given by 
\begin{equation} \label{Lie_bracket_W}
\begin{aligned}
& \textstyle \{ \int \phi, \int \psi\}_1^{\mathcal{W}}= \int\{ \phi_\chi \psi\}_1^{\mathcal{W}}|_{\chi=0}=- \int \big(\sum_{t\in I} q_t \otimes \frac{\delta \phi}{\delta \bar{q}_t}| [ \mathcal{L}_u, \sum_{k\in I} q_{k} \otimes \frac{\delta{\psi}}{\delta{\bar{q}_{k}}}] \big), \\
& \textstyle \{ \int \phi, \int \psi\}_2^{\mathcal{W}}= \int\{ \phi_\chi \psi\}_2^{\mathcal{W}}|_{\chi=0}= \int \big(\sum_{t\in I} q_t \otimes \frac{\delta \phi}{\delta \bar{q}_t}| [ s\otimes 1, \sum_{k\in I} q_{k} \otimes \frac{\delta{\psi}}{\delta{\bar{q}_{k}}}] \big).
\end{aligned}
\end{equation}
for any $\phi, \psi \in \mathcal{W}(\g, \n, \mathfrak{m}, f)$.
\end{proposition}

\begin{proof}
We note that since $f \notin \mathfrak{b}$ we have 
$f=\sum_{k \in \tilde{I} \setminus I} (f|q_k) q^k,$
which implies that  
\begin{equation*}
\textstyle \sum_{k\in \tilde{I}} q^k\otimes \bar{q}_k= \sum_{k\in {I}} q^k\otimes \bar{q}_k+f \otimes 1, \quad (\text{mod}\ \g \otimes \mathcal{I}_{f, \mathfrak{m}}).
\end{equation*}
It follows from  Lemma \ref{Lemma:Lax operator-affine} that for any $a,b \in \mathcal{P}$ one has 
\begin{equation}\label{auxa}
 \textstyle \{ \smallint a, \smallint b\}_1=- \smallint (\sum_{t\in I} q_t \otimes \frac{\delta a}{\delta \bar{q}_t}| [ \mathcal{L}_u, \sum_{k\in I} q_{k} \otimes \frac{\delta{b}}{\delta{\bar{q}_{k}}}]), \quad  (\text{mod} \, \, \, \mathcal{I}_{f, \mathfrak{m}}/D \mathcal{I}_{f, \mathfrak{m}}).
\end{equation}
In the sequel of the proof we use the direct sum decomposition of $\mathcal{P}(\overline \g)$
\begin{equation} \label{directsum} 
\mathcal{P}(\overline \g)=\mathcal{P} \oplus \mathcal{I}_{f, \mathfrak{m}}.
\end{equation}
which implies 
$\mathcal{P}(\overline \g)/D\mathcal{P}(\overline \g) =\mathcal{P}/D\mathcal{P}  \oplus \mathcal{I}_{f, \mathfrak{m}}/D\mathcal{I}_{f, \mathfrak{m}}.$

From the definition \eqref{bracket_H} of the first $\chi$-bracket on $\mathcal{W}(\g, \n, \mathfrak{m}, f)$ we have, for all $\phi, \psi \in \mathcal{W}(\g, \n, \mathfrak{m}, f)$,
\begin{equation}\label{auxb}
\{ \smallint \phi, \smallint \psi \}_1-\{ \smallint \phi, \smallint \psi \}_1^{\mathcal{W}} \in \mathcal{I}_{f, \mathfrak{m}}/D \mathcal{I}_{f, \mathfrak{m}}, \quad 
\{ \smallint \phi, \smallint \psi \}_1^{\mathcal{W}} \in  \mathcal{P}/D\mathcal{P}.
\end{equation}
Since the RHS of \eqref{auxa} lies in $\mathcal{P}/D\mathcal{P}$, we conclude after comparing equations \eqref{auxa} and \eqref{auxb} that,  for all $\phi, \psi \in \mathcal{W}(\g, \n, \mathfrak{m}, f)$,
\begin{equation}
 \textstyle \{ \smallint \phi, \smallint \psi\}_1^{\mathcal{W}}=- \smallint (\sum_{t\in I} q_t \otimes \frac{\delta \phi}{\delta \bar{q}_t}| [ \mathcal{L}_u, \sum_{k \in I} q_{k} \otimes \frac{\delta{\psi}}{\delta{\bar{q}_{k}}}]).
\end{equation}
The second line  of \eqref{Lie_bracket_W} follows immediately from Lemma \ref{Lemma:Lax operator-affine}.
\end{proof}

\begin{definition}\label{extbrabil}
 Let $\hat{\g}:= \g(\!(z^{-1})\!)$.
Then $\mathbb{C}[D] \ltimes  (\hat{\g} \otimes \mathcal{P}(\overline{\g}))$ is a Lie superalgebra endowed with the Lie bracket and bilinear form defined by 
\begin{equation}
\begin{aligned}
  &[az^m \otimes u, b z^n \otimes v] = (-1)^{p(b)p(u)} [a,b]z^{m+n} uv, \\
 &[D, a z^m \otimes u] = (-1)^{p(a)} a z^m \otimes u', \\
 &(a z^m \otimes u|b z^n \otimes v)= (-1)^{p(b) p(u)} \delta_{m+n,0}(a|b) uv.
\end{aligned}
\end{equation}
\end{definition}

Recall that the element $s$ is in $\text{ker} \, \text{ad} \, \n$. We let
\[ \mathcal{L}_s := \mathcal{L}_u+  sz \otimes 1 \in \mathbb{C}[D] \ltimes ( \hat{\g} \otimes \mathcal{P}).\]

\begin{proposition} \label{Prop: Lax_bracket_W}
For all $\phi, \psi \in \mathcal{W}(\g, \n, \mathfrak{m}, f)$ we have
\begin{equation} \label{Eqn:compatible bracket-Lax}
\begin{aligned}
& \textstyle\{\int \phi, \int \psi\}_1^{\mathcal{W}}=- \int (\sum_{t\in I} q_t \otimes \frac{\delta \phi}{\delta \bar{q}_t}| [ \mathcal{L}_s, \sum_{k\in I} q_{k} \otimes \frac{\delta{\psi}}{\delta{\bar{q}_{k}}}]),\\
& \textstyle \{\int \phi, \int \psi\}_2^{\mathcal{W}}= \int (\sum_{t\in I} q_t \otimes \frac{\delta \phi}{\delta \bar{q}_t}| [ z^{-1}\mathcal{L}_s, \sum_{k\in I} q_{k} \otimes \frac{\delta{\psi}}{\delta{\bar{q}_{k}}}]).
\end{aligned}
\end{equation}
\end{proposition}
\begin{proof}
Immediate by Definition \ref{extbrabil} and Proposition \ref{pro}.
\end{proof}

\section{SUSY Drinfeld-Sokolov hierarchies} \label{Sec:DS hier}
We keep the notations and assumptions (A-1) $\cdots$ (A-4) of Section $4$. Recall that $\{q^t |  t \in \tilde{I} \}$ is a homogeneous basis of $\g$ such that $\mathfrak{b}=\text{Span}_{\mathbb{C}}\{q^t| t \in I\}$, $[f,\n]=\text{Span}_{\mathbb{C}}\{q^t| t \in J\}$ and $V=\text{Span}_{\mathbb{C}}\{q^t| t \in I \setminus J\}$, where $J \subset I \subset \tilde{I} \subset \Z$. This basis is dual to $\{q_t| t \in \tilde{I} \}$, i.e. $(q^t|q_{t'})= \delta_{t,t'}$, and $\text{Span}_{\mathbb{C}}\{q_t| t \in I\}=\mathfrak{b}_-$.  

Recall  $s\in \text{ker} \, \text{ad} \, \mathfrak{n}$ is a homogeneous odd element  with degree $j$ and $f \in \g_{-i}$. We extend the gradation of $\mathfrak{g}$ to $\hat{\mathfrak{g}}$ by giving $z$ the grading $-i-j$ and denote the $k$-th component of $\hat{\mathfrak{g}}$ by $\hat{\mathfrak{g}}_k$. Finally, we introduce 
\begin{equation} \label{lambda}
 \Lambda:=f+sz \in \hat{\mathfrak{g}}_{-i}
 \end{equation}
 and assume that
 \begin{enumerate}
 \item[(A-5)]
$ \text{ad} \hspace{1 mm} \Lambda^2$ is a semisimple element in $\hat{\mathfrak{g}}$, i.e. we have the direct sum decomposition 
 \begin{equation} \label{Eqn:semisimple}
 \hat{\mathfrak{g}}= \mathcal{K} \oplus \mathcal{I}=   \text{ker} \, \text{ad} \, \Lambda^2 \oplus \text{im}\, \text{ad} \, \Lambda^2.
  \end{equation}
  \end{enumerate}
  \if
 where $\mathcal{K}:=\text{ker} \, \text{ad} \, \Lambda^2$ and $ \mathcal{I}:=\text{im} \, \text{ad} \, \Lambda^2$. 
 \fi
 We will use the notation 
\[ a=a_{\mathcal{K}}+a_{\mathcal{I}}\]
 to decompose elements in $\hat{\g}$ according to \eqref{Eqn:semisimple}. Let $\mathcal{Z}$ be the center of $\mathcal{K}$. It is non trivial as it contains the even elements $z^k \Lambda^{2n}$ where $k, n \in \Z$. We checked that for $\g=\mathfrak{sl}(p|q)$ no matter the choice of $\Lambda$, this center is always even. We do not know if this fact is true for all (simple) Lie superalgebras and therefore we will state our results in full generality.


\subsection{Examples} \label{Subsec:Example-(g,f,m,n,s)}
We give here several families of examples of Lie superalgebras $\g$, as well as a choice of $\n, \mathfrak{m}, f,$ and $s$ that satisfy the assumptions (A-1),$\cdots$,(A-5). We were not able to find an example where $\text{dim} \, V=\text{dim} \, \mathfrak{b}-\text{dim} \, \n$ is stricly smaller than the dimension of $\g_{\bar{0}}$, the even part of $\g$. Note that for the three families of simple Lie superalgebras below, the choice of $\n, \mathfrak{m}, f,$ and $s$ satisfying the assumptions (A-1),$\cdots$,(A-5) is relatively large. In each case we pick a tuple $(\n, \mathfrak{m}, f, s)$ so that $\text{dim} \, V=\text{dim} \, \g_{\bar{0}}$.

\begin{example}
The Lie superalgebra $\g = \mathfrak{osp}(2n|2n)$ can be realized as the subalgebra of $\mathfrak{gl}(2n|2n)$ consisting of matrices of the form
\begin{equation} \label{matirx(2n|2n)}
A=\left(\begin{array}{cc|cc} A_{11} & A_{12}& A_{42}^t & A_{32}^t \\ A_{21} & -A_{11}^t & -A_{41}^t & -A_{31}^t \\ \hline  A_{31} & A_{32} & A_{33} & A_{34} \\  A_{41} & A_{42} & A_{43} & -A_{33}^t \end{array} \right),
\end{equation}
where  $A_{12},A_{21}$ (resp. $A_{34},A_{43}$) are $n\times n$ symmetric (resp. skewsymmetric) matrices.
 Let us define the degree of each block matrix $A_{ij}$ as follows:
\begin{equation}
\left(\begin{array}{cc|cc} 0 & 1& 0 & 1 \\ -1 & 0 & -1 & 0 \\ \hline  0 & 1 & 0 & 1 \\  -1 & 0 & -1 & 0 \end{array} \right).
\end{equation}
Then $\g= \g_{-1} \oplus \g_0 \oplus \g_1$.
 Take

\begin{equation*}
\Lambda=f+zs = \left(\begin{array}{cc|cc} 0 & 0&0 & z I_n \\ 0 & 0 & -I_n & 0 \\ \hline  0 & z I_n & 0 & 0 \\  I_n & 0 & 0 & 0 \end{array} \right), \, \, \,  \text{so that} \, \, \,  \Lambda^2= z \left(\begin{array}{cc|cc}  I_n & 0&0 & 0 \\ 0 & -I_n & 0 & 0 \\ \hline  0 & 0 & -I_n & 0 \\  0 & 0 & 0 & I_n \end{array} \right) 
\end{equation*}
and let \[ \mathfrak{n}=\mathfrak{m}= \g_1.\]
One can check that the tuple $(\g, \n, \mathfrak{m}, f)$ satisfies the assumptions (A-1) to  (A-5). Note that  $\text{dim} \, V=\text{dim} \, \g_{\bar{0}}$.
\end{example}
 
\begin{example}
Let $\mathfrak{g}= \mathfrak{sl}(m|n)$ for $m>n$. We split matrices in $9$ blocks (for sizes see the matrix $\Lambda$ below) and pick
\[\Lambda=f+zs = \left(\begin{array}{cc|c} 0 & 0 & z I_n  \\   0 & 0 & 0  \\ \hline  I_n & 0  & 0 \end{array} \right), \, \, \, \text{so that} \, \, \,  \Lambda^2= z \left(\begin{array}{cc|c} I_n & 0 & 0  \\   0 & 0 & 0  \\ \hline  0 & 0  & I_n \end{array} \right). 
\]
We consider the principal grading by block, that is to say 
\[ \left(\begin{array}{cc|c} 0 & 1 & 2   \\   -1 & 0 & 1  \\ \hline  -2 & -1  & 0 \end{array} \right) 
\]
and we let $\mathfrak{n}=\g_{ \geq 1}$ and $\mathfrak{m}=\g_{\geq 2}$.
The tuple $(\g, \n, \mathfrak{m}, f)$ satisfies assumptions (A-1) to  (A-5). Moreover in this case $\text{dim} \, V=\text{dim} \, \g_{\bar{0}}$.
\end{example}

\begin{example}
Let $\g=\mathfrak{sl}(n|n)/I_{n|n}$.  We consider the grading 
\begin{equation}
\left(\begin{array}{cc} 0 & 1 \\ -1& 0\end{array} \right)
\end{equation}
and we pick 
\begin{equation}
\Lambda=f+zs=\left(\begin{array}{cc} 0 & zA \\ I_n& 0\end{array} \right),\, \, \, \text{so that} \, \, \, \Lambda^2=\left(\begin{array}{cc} zA & 0 \\ 0& zA\end{array} \right),
\end{equation}
where $A$ is any diagonalizable matrix. It is straightforward to check that $\text{ad } \, \Lambda^2$ is semisimple. We let $\n=\mathfrak{m}=\g_{1}$. All conditions (A-1) to  (A-5)  are met. Note that in this example also we have $\text{dim} \, V=\text{dim} \, \g_{\bar{0}}$.
\end{example}

\subsection{Super Hamiltonian Systems on $\mathcal{P}$}\label{Subsec:main} In this section, we construct a family $(\smallint h_C)_{C \in \mathcal{Z}}$ of functionals that pairwise commute for the second Lie superalgebra bracket on $\mathcal{W}(\g, \n, \mathfrak{m}, f)/ D \mathcal{W}(\g, \n, \mathfrak{m}, f)$, from which we construct a hierarchy of compatible super Hamiltonian equations on $\mathcal{P}$ (Theorem \ref{bigthe}).
\\
\indent
For any $\mathcal{Q} \in \mathcal{F}$, we let 
\begin{equation*}
\mathcal{L}_s(\mathcal{Q})=D+\mathcal{Q}+ \Lambda \otimes 1=\mathcal{L}(\mathcal{Q})+z s \otimes 1 \in \mathbb{C}[D] \ltimes( \hat{\g} \otimes \mathcal{P)}.
\end{equation*}
We simply write 
\[ \textstyle \mathcal{L}_s:= \mathcal{L}_s(\mathcal{Q}_u), \quad \mathcal{L}^c:= \mathcal{L}_s(\mathcal{Q}^c_u)\]
where $\mathcal{Q}_u=\sum_{t \in I} {q^{t} \otimes \bar{q}_t}$ and $\mathcal{Q}_u^c=\sum_{t \in I \setminus J} {q^{t} \otimes w_t}$.

\begin{lemma} \label{first lemma_DS}
Let $\mathcal{Q}=\sum_{t \in I}{q^t \otimes a_t} \in \mathcal{F}$ and let $\mathcal{P}_{\mathcal{Q}} \subset \mathcal{P}$ be the differential superalgebra generated by the $a_t$'s. There exists a unique pair $(T_{\mathcal{Q}},H_{\mathcal{Q}}) \in (\mathcal{I} \otimes \mathcal{P}_{\mathcal{Q}}) \times (\mathcal{K} \otimes \mathcal{P}_{\mathcal{Q}})$ such that
\begin{equation} \label{TH}
e^{ad \hspace{1 mm} T_{\mathcal{Q}}} [ \mathcal{L}_s(\mathcal{Q}), \mathcal{L}_s(\mathcal{Q})]/2=D^2+H_{\mathcal{Q}}+\Lambda^2\otimes 1.
\end{equation}
\end{lemma}
\begin{proof}
We write 
\[ [\mathcal{L}_s(\mathcal{Q}), \mathcal{L}_s(\mathcal{Q})]/2=D^2+\mathcal{U}+\Lambda^2 \otimes 1\] and look for $T_{\mathcal{Q}}=T_{1}+T_{2}+\cdots$ and $H_{\mathcal{Q}}=H_{-2i+1}+H_{-2i+2}+\cdots$ satisfying \eqref{TH}. The $\hat{\g}_{-2i+1}\otimes \mathcal{P}_{\mathcal{Q}}$-part of equation \eqref{TH} reads 
\[\mathcal{U}_{-2i+1}=H_{-2i+1}+ [\Lambda^2 \otimes 1,T_{1}].\]  

By \eqref{Eqn:semisimple}, there exists a unique solution $(H_{-2i+1},T_{1})$ to this equation in $(\mathcal{K} \otimes \mathcal{P}_{\mathcal{Q}}) \times (\mathcal{I} \otimes \mathcal{P}_{\mathcal{Q}})$. 
Inductively, we can find $H_{-2i+k}$ and $T_{k}$ using the $\hat{\g}_{-2i+k}\otimes \mathcal{P}_{\mathcal{Q}}$-part of \eqref{TH} since the part of  \eqref{TH} can be reformulated as 
\[ H_{-2i+k}+[\Lambda^2 \otimes 1, T_{k}]=\text{ an expression  depending  only on $\mathcal{U}$,  $H_{<-2i+k}$ and $T_{<k}$}.\]
Such an equation can be solved uniquely for $(H_{-2i+k},T_{k}) \in (\mathcal{K} \otimes \mathcal{P}_{\mathcal{Q}}) \times (\mathcal{I} \otimes \mathcal{P}_{\mathcal{Q}})$ thanks to \eqref{Eqn:semisimple}.
\end{proof}

When $\mathcal{Q}=\mathcal{Q}_u$  (resp. $\mathcal{Q}=\mathcal{Q}_u^c$) we simply denote the solution to \eqref{TH} by $(T,H)$ (resp. $(T^c,H^c)$), i.e., 
\begin{equation} \label{TcHc notation}
T:= T_{\mathcal{Q}_u}, \quad  T^c:= T_{\mathcal{Q}^c_u}, \quad H:= H_{\mathcal{Q}_u}, \quad H^c:= H_{\mathcal{Q}^c_u}.
\end{equation}

\begin{lemma}  \label{second lemma_DS}
Let $K_{\mathcal{Q}}$ be such that 
\begin{equation} \label{TH2}
e^{ad \,  T_{\mathcal{Q}}} \mathcal{L}_s({\mathcal{Q}})=D_{\mathcal{Q}}+K_{\mathcal{Q}}+\Lambda \otimes 1,
\end{equation}
where $(T_{\mathcal{Q}},H_{\mathcal{Q}})$ solves \eqref{TH}. Then $K_{\mathcal{Q}} \in \mathcal{K} \otimes \mathcal{P}_{\mathcal{Q}}$.
\end{lemma}
\begin{proof}
Let us show, for  $X \in \hat{\g} \otimes \mathcal{P}$, if  $ [\Lambda \otimes 1, X] \in \mathcal{K} \otimes \mathcal{P}$ then $X \in \mathcal{K}  \otimes \mathcal{P}$. Decompose $X=Y+[\Lambda^2 \otimes 1, Z]$ with $Y \in \mathcal{K}  \otimes \mathcal{P}$. We have $[\Lambda \otimes 1, X]=[\Lambda \otimes 1, Y]+[\Lambda \otimes 1, [\Lambda^2 \otimes 1, Z]] \in \mathcal{K}$. Since $[\Lambda \otimes 1, Y] \in \mathcal{K}  \otimes \mathcal{P}$ by the Jacobi identity, we deduce that $[\Lambda \otimes 1, [\Lambda^2 \otimes 1, Z]]=0$. Hence $ [\Lambda^2 \otimes1 , Z] \in \text{ker} \, \text{ad} \, \Lambda   \otimes   \mathcal{P} \subset \mathcal{K} \otimes \mathcal{P}$. Therefore $[\Lambda^2 \otimes 1,Z]=0$ and $X=Y \in \mathcal{K}  \otimes \mathcal{P}$.

Since $[D+ K_{\mathcal{Q}}+\Lambda\otimes 1, D+ K_{\mathcal{Q}}+\Lambda\otimes 1]/2= D^2 + H_{\mathcal{Q}}+ \Lambda^2\otimes 1$, we have
\begin{equation} \label{KH}
\textstyle H_{\mathcal{Q}}=K_{\mathcal{Q}}'+[K_{\mathcal{Q}},K_{\mathcal{Q}}]/2+[\Lambda\otimes 1, K_{\mathcal{Q}}].
\end{equation}
To see the lemma, let us use induction on the graded parts of $K_{\mathcal{Q}}$. The $\hat{\g}_t \otimes \mathcal{P}_{\mathcal{Q}}$-part of \eqref{KH} is
\begin{equation} \label{KH-1}
\textstyle (H_{\mathcal{Q}})_t=(K_{\mathcal{Q}})'_t+ \sum_{t'<i} [(K_{\mathcal{Q}})_{-t'}, (K_{\mathcal{Q}})_{t+t'}]/2+ [\Lambda\otimes 1, (K_{\mathcal{Q}})_{t+i}].
\end{equation}
By Lemma \ref{first lemma_DS}, the LHS of \eqref{KH-1} is in $\mathcal{K} \otimes \mathcal{P}_{\mathcal{Q}}$ for any $t >-2i$. If $t=-2i+1$, only the last term in the RHS of \eqref{KH-1} is nontrivial, and we deduce from the first part of the proof that $(K_{\mathcal{Q}})_{-i+1} \in \mathcal{K} \otimes \mathcal{P}_{\mathcal{Q}}$.
If $(K_{\mathcal{Q}})_{t+t'}$ is in $\mathcal{K} \otimes \mathcal{P}_{\mathcal{Q}}$ for any $t'<i$, it follows from \eqref{KH-1} that $[\Lambda\otimes1, (K_{\mathcal{Q}})_{t+i}]\in \mathcal{K} \otimes \mathcal{P}_{\mathcal{Q}}$. This implies that $(K_{\mathcal{Q}})_{t+i} \in \mathcal{K} \otimes \mathcal{P}_{\mathcal{Q}}$ and thus we proved that  $K_{\mathcal{Q}} \in \mathcal{K} \otimes \mathcal{P}_{\mathcal{Q}}$.
\end{proof}
Following \eqref{TcHc notation} we write
\begin{equation} \label{Kc notation}
K:= K_{\mathcal{Q}_u}, \quad  K^c:= K_{\mathcal{Q}^c_u}.
\end{equation}

\begin{lemma} \label{KW} 
There exists an even element $ \widetilde{T}\in \hat{\g} \otimes \mathcal{P}$ such that
\begin{equation} \label{kw}
e^{\text{ad} \hspace{1 mm} \widetilde{T}} \mathcal{L}_s=D+K^c+\Lambda\otimes 1\, . 
\end{equation}
\end{lemma}
\begin{proof}
By Lemma \ref{second lemma_DS}, we have
\begin{equation}
 e^{\text{ad} T^c} \mathcal{L}^c = D+ K^c + \Lambda \otimes 1.
 \end{equation}
Let $N$ be the unique element in $(\mathfrak{n} \otimes \mathcal{P})_{\bar{0}}$ such that $ \mathcal{L}_s = e^{\text{ad} N} \mathcal{L}^c$. We have
\[ \textstyle e^{\text{ad} T^c} \mathcal{L}^c = e^{\text{ad} T^c} e^{-\text{ad} N}(\mathcal{L}_{s})=D + K^c +\Lambda \otimes 1.\]
Let $\widetilde{T} \in \hat{\g} \otimes \mathcal{P}$ be such that $e^{\text{ad} \widetilde{T}} = e^{\text{ad} T^c} e^{-\text{ad}N}$. Such $\tilde{T}$ exists since both $T^c$ and $N$ have positive gradings. The pair $(\tilde {T}, K^c)$ satisfies
\begin{equation} \label{KW_eqn}
 e^{\text{ad} \widetilde{T}}  \mathcal{L}_s =D+K^c+\Lambda \otimes 1,
 \end{equation} 
which concludes the proof.
\end{proof}

\begin{remark}
In Lemma \ref{KW}, $\widetilde{T}$ need not be in $\mathcal{I} \otimes \mathcal{P}$. However, $ T^c$ in \eqref{KW_eqn} is an element of $\mathcal{I} \otimes \mathcal{W}(\g,\n,\mathfrak{m},f)$.
\end{remark}

\begin{lemma} \label{unicity}
\[(K-K^c) \in \mathcal{K}\otimes \partial \mathcal{P}+ [\mathcal{K}, \mathcal{K}]\otimes \mathcal{P}.\] 
\end{lemma}
\begin{proof}
Let $U \in \hat{\g} \otimes \mathcal{P}$ be such that $e^{ad \, U}=e^{ad \, T}e^{-ad \, \tilde{T}}$. Such  $U$ exists since both $T$ and $\tilde{T}$ have positive gradings. It follows from Lemma \ref{second lemma_DS} and Lemma \ref{KW} that
\begin{equation}
e^{ad \hspace{1 mm} U} (D+K^c+\Lambda\otimes 1)=D+K+\Lambda\otimes 1.
\end{equation}
 A similar argument as in the proof of Lemma \ref{second lemma_DS} shows  that $U \in \mathcal{K}\otimes \mathcal{P}$. Next we write
 \begin{equation}
 K-K^c=(e^{ad \, U}-1)(D+K^c+\Lambda \otimes 1).
 \end{equation}
 Since $U \in \mathcal{K}\otimes \mathcal{P}$, it is clear that  $(e^{ad \hspace{1 mm} U}-1)D$ is in $\mathcal{K}\otimes \partial \mathcal{P}+ [\mathcal{K}, \mathcal{K}]\otimes \mathcal{P}$ and that $(e^{ad \hspace{1 mm} U}-1)(K^c+\Lambda \otimes 1)$ is in $[\mathcal{K}, \mathcal{K}]\otimes \mathcal{P}$. Hence $(K-K^c) \in \mathcal{K}\otimes \partial \mathcal{P}+ [\mathcal{K}, \mathcal{K}]\otimes \mathcal{P}$.
\end{proof}

We are now ready to define the functionals $\smallint h_C, C \in \mathcal{Z},$ on the generalized $W$-algebra  $\mathcal{W}(\g,\n,\mathfrak{m},f)$. These will be the hamiltonians of our super Hamiltonian equations \eqref{hamiltequa}, which pairwise supercommute (Theorem \ref{bigthe}).

\begin{definition}
Let $C \in \mathcal{Z}$. We define 
\begin{equation}
h_C=(K^c| C \otimes 1) \in  \mathcal{W}(\g,\n,\mathfrak{m},f).
\end{equation}
\end{definition}

\begin{corollary} \label{cor}
Let  $C \in \mathcal{Z}$.
\[
\smallint h_{C}=\smallint (K| C \otimes 1).
\]
\end{corollary}
\begin{proof}
Immediate from Lemma \ref{KW} and Lemma  \ref{unicity}, after using the invariance of the bilinear form $( \ | \ )$.
\end{proof}
In practice, it is often enough to compute the quadratic part of $\smallint h_C$ as a super differential polynomial in $ \mathcal{W}(\g,\n,\mathfrak{m},f)$ to deduce the linear independence of the family $(\smallint h_{z^nC})_{n \geq 0}$. The following proposition gives a formula to compute this quadratic part.
\begin{proposition} \label{indept} Recall the definition \eqref{TcHc notation} of $T^c$. Let $C \in \mathcal{Z}$.	
The linear and quadratic parts of  $ \smallint h_{z^nC}$  for $n \geq 0$ are given by
\begin{equation}
 \int (\mathcal{Q}_{\mathcal{K}}^c | C \otimes z^n) + \frac{1}{2} \int (\mathcal{Q}_{\mathcal{I}}^c|[C \otimes z^n, T_{l}]),
\end{equation}
where $T_l$ denotes the linear part of $T^c$.
\end{proposition}

\begin{proof}
 We denote by $K_l$ and $T_l$ (resp. $K_q$ and $T_q$) the linear (resp. quadratic) part of $K^c$ and $T^c$. The linear part of equation \eqref{TH2} where $\mathcal{Q}=\mathcal{Q}^c$ splits into two ways using \eqref{Eqn:semisimple}:
\begin{equation*}
\begin{split}
\mathcal{Q}_{\mathcal{I}}^c &=T_l'+[\Lambda \otimes 1, T_l], \\
\mathcal{Q}_{\mathcal{K}}^c&=K_l.
\end{split}
\end{equation*}
We now move on to the quadratic part of $K$. It satisfies the equation
\begin{equation*} 
\begin{split}
K_q&=-T_q'-\frac{1}{2}[T_l,T_l']-[\Lambda \otimes 1,T_q]+\frac{1}{2}[T_l,[T_l, \Lambda \otimes 1]]+[T_l, \mathcal{Q}^c] \\
      &=-T_q'-[\Lambda \otimes 1,T_q]+\frac{1}{2}[T_l, \mathcal{Q}_{\mathcal{I}}^c+2\mathcal{Q}_{\mathcal{K}}^c]
      \end{split}
\end{equation*}
Since both $[\Lambda \otimes 1,T_q]$ and $[T_l,\mathcal{Q}_{\mathcal{K}}^c]$ are in $\mathcal{I}$, we get 
\begin{equation*} 
(h_{z^nC})_q=-(T_q|C \otimes z^n)'+\frac{1}{2}([T_l, \mathcal{Q}_{\mathcal{I}}^c]| C \otimes z^n), \, \, n=0,1,2,\cdots
\end{equation*}
by invariance of the bilinear form. 
Therefore for all, $n \geq 0$,
\begin{equation*}
\smallint (h_{z^nC})_q=\smallint  \frac{1}{2}(\mathcal{Q}_{\mathcal{I}}^c|[C \otimes z^n, T_l]).
\end{equation*}
\end{proof}

\begin{definition}
Let $C \in \mathcal{Z}$. We define the following element of $\hat{\g} \otimes \mathcal{P}$
\begin{equation} \label{MC}
M_C=e^{-ad \, T} (C \otimes 1).
\end{equation}
$M_C$ splits as the sum of its holomorphic and meromorphic parts
\begin{equation*}
M_C^+=\sum_{ n \geq 0} M_C^n z^n, \, \, \, \, M_C^-=\sum_{ n < 0} M_C^n z^n.
\end{equation*}
We denote by $(M_C^n)_{\mathfrak{b}_-}$
the projection of $M_C^n\in \g \otimes \mathcal{P}$ onto the space $\mathfrak{b}_- \otimes \mathcal{P}$ according to the decomposition
\begin{equation*}
\g \otimes \mathcal{P}= (\mathfrak{b}_- \otimes \mathcal{P}) \oplus (\mathfrak{m} \otimes \mathcal{P}). 
\end{equation*}
\end{definition}

It is clear from Lemma \ref{KW} that 
\begin{equation} \label{LMCcom}
[M_C, \mathcal{L}_s]=0 \, \, \, \text{for all} \, \, C \in \mathcal{Z}.
\end{equation}
In the next two lemmas, we relate $M_C$ to the variational derivative of $h_C$.

\begin{lemma}\label{Lemma for Lemma h_C}
Let $U,V\in \g \otimes \mathcal{P}$ with $U$ even. For any $\mathcal{Q} \in \mathcal{F}$, we have
\begin{equation}
\sum_{h \geq 0} [\text{ad }U^h (V) , e^{\text{ad } T}(\mathcal{L}_s(\mathcal{Q}))] = \sum_{h,k \geq 0} \frac{1}{(h+k+1)!} \text{ad } U^h \text{ad } V \text{ad }U^k (\mathcal{L}_s(\mathcal{Q})).
\end{equation}
\end{lemma}
\begin{proof}
See Lemma 2.6 in \cite{DKV13}.
\end{proof}

\begin{lemma} \label{Lemma: h_C}
Let $C \in \mathcal{Z}$.Then
\[{  (M_C^0)_{\mathfrak{b}_-}}= \sum_{k\in I} q_k \otimes \frac{\delta h_C}{\delta \bar{q}_k}\, .\]
\end{lemma}

\begin{proof}  
We extend the definition of the variational derivative to $\g \otimes \mathcal{P}$ by letting 
\begin{equation*}
\frac{\delta}{\delta \bar{q}_k}(a \otimes u) :=(-1)^{p(a)(k+1)} a \otimes \frac{\delta u}{\delta \bar{q}_k} \, \,  \text{for} \, \, a \otimes u \in \g\otimes \mathcal{P}.
\end{equation*}  By Lemma \ref{second lemma_DS}, 
\begin{equation} \label{lemma_5.9_1}
\begin{aligned}
 \textstyle \frac{\delta K}{\delta \bar{q}_k}& = \textstyle \frac{\delta}{\delta \bar{q}_k} \left( e^{ad T} \mathcal{L}_s-D\right)\\
& =\textstyle \frac{\delta \mathcal{Q}_u}{\delta \bar{q}_k}   + \sum_{m\geq 0}(-1)^{m(k+1)+\frac{m(m+1)}{2}} \sum_{t\geq 1} \frac{1}{t!} D^m \frac{\partial}{\partial \bar{q}_{k}^{(m)}}\big( ad T^t (\mathcal{L}_s)\big)\\
\end{aligned}
\end{equation}
We have 
\begin{equation}\label{lemma_5.9_2}
 \textstyle \frac{\delta}{\delta \bar{q}_k} \mathcal{Q}_u= q^k \otimes 1,
 \end{equation}
 and
 \begin{equation}\label{lemma_5.9_3}
\begin{aligned}
& \textstyle   \frac{\partial}{\partial \bar{q}_{k}^{(m)}}\big( ad T^t (\mathcal{L}_s-D)\big) \\
& = \textstyle  \sum_{t'=0}^{t-1}
( ad T)^{t'}  (ad \frac{\partial T}{\partial \bar{q}_{k}^{(m)}}) (adT)^{t-1-t'} (\mathcal{Q}_u+\Lambda \otimes 1)+\delta_{m,0} (ad T)^t(q_k \otimes 1).
\end{aligned}
\end{equation}
By Lemma \ref{Lemma: commutator_derivatives}, we also have
\begin{equation}\label{lemma_5.9_4}
\begin{aligned}
 \textstyle \frac{\partial}{\partial  \bar{q}_{k}^{(m)}}\big( ad T^t (D)\big) & \textstyle=  \sum_{t'=0}^{t-2}
( ad T)^{t'}  (ad \frac{\partial T}{\partial  \bar{q}_{k}^{(m)}}) (adT)^{t-2-t'} (-T') \\
& \qquad  \textstyle  + (ad T)^{t-1} \left( -\frac{\partial T}{\partial  \bar{q}_{k}^{(m-1)}} +  (-1)^{k+m}(\frac{\partial}{\partial \bar{q}_k^{(m)}}T)' \right),
\end{aligned}
\end{equation}
where by definition $\frac{\partial}{\partial \bar{q}_{k}^{(-1)}}=0$.

After comparing \eqref{lemma_5.9_1}, \eqref{lemma_5.9_2}, \eqref{lemma_5.9_3} and \eqref{lemma_5.9_4} we deduce that 
\begin{equation} \label{lemma_5.9_5}
\begin{aligned}
 \frac{\delta K}{\delta \bar{q}_k}\textstyle  & = \textstyle  e^{ad T}(q^k \otimes 1)  +\sum_{m\geq 0} (-1)^{m(k+1)+ \frac{m(m+1)}{2}}  \\
& \quad  \textstyle D^m \left( \sum_{t\geq 1} \sum_{t'=0}^{t-1}\frac{1}{t!}
( ad T)^{t'}  (ad \frac{\partial T}{\partial {\bar{q}_k}^{(m)}}) (adT)^{t-1-t'} \mathcal{L}_s \right)
 \\
&  \quad \textstyle  -\sum_{m\geq 1}\sum_{t\geq 1} (-1)^{m(k+1)+\frac{m(m+1)}{2}}D^m \left( (ad T)^{t-1} (\frac{\partial T}{\partial {\bar{q}_k\, (m-1)} })\right).
\end{aligned}
\end{equation}
Note that  by Lemma \ref{Lemma for Lemma h_C}, 
\begin{equation*}
\textstyle 
\sum_{t\geq 1} \sum_{t'=0}^{t-1}\frac{1}{t!}
( ad T)^{t'}  (ad \frac{\partial T}{\partial {\bar{q}_k}^{(m)}}) (adT)^{t-1-t'} \mathcal{L}_s =\sum_{h\geq 1}\frac{1}{h !}\left[ (ad T)^{h-1} \left(\frac{\partial T}{\partial {\bar{q}_k}^{(m)}}\right), e^{ad T} (\mathcal{L}_s)\right.
]
\end{equation*}
Therefore \eqref{lemma_5.9_5} can be rewritten as 

\begin{equation}
\begin{aligned}
 \frac{\delta K}{\delta \bar{q}_k}& = \textstyle  e^{ad T}(q^k \otimes 1) + \sum_{m\geq 0} (-1)^{m(i+1)+ \frac{m(m+1)}{2}} \\
&   \quad \cdot  \textstyle D^m \left( \sum_{h\geq 1}\frac{1}{h !}\left[ (ad T)^{h-1} \left(\frac{\partial T}{\partial {\bar{q}_i}^{(m)}}\right), e^{ad T} (\mathcal{L}_s) -D\right] \right).
\end{aligned}
\end{equation}

Since $[ (e^{ad T} \mathcal{L}_s-D)^{(l)}, C\otimes 1]=0$ for any $l \geq 0$ we get by invariance of the bilinear form
\begin{equation}
\begin{aligned}
 \frac{\delta h_C}{\delta\bar{q}_k}=  (  e^{ad T}(q^k \otimes 1)|C \otimes 1)=(  q^k \otimes 1| e^{ad -T}(C\otimes 1) )= (q^k \otimes 1|M_C)
\end{aligned}
\end{equation}
which conludes the proof.
\end{proof}

We will abuse notations using the following definition and use $d/dt_C$ to denote both the derivation $\{ \smallint h_C {}_{\chi} \bullet \}_2 |_{\chi=0}$ of $\mathcal{P}(\overline \g)$ and its extension to $\mathfrak{g} \otimes \mathcal{P}(\overline \g)$, for any $C \in \mathcal{Z}$.

\begin{definition}\label{extder}
Let $d$ be a derivation of $\mathcal{P}(\overline \g)$. We extend $d$ to a map $\hat{d}$  from $\g \otimes \mathcal{P}(\overline \g)$ to itself by letting, for all $t \in \tilde I$ and $p_t \in \mathcal{P}(\overline \g)$ 
\begin{equation} \label{extder}
\hat{d} (q^t \otimes p_t)=(-1)^{t\, p(d)}q^t \otimes d(p_t).
\end{equation}
$d$ is a derivation of the Lie superalgebra $\g \otimes \mathcal{P}(\overline \g)$. Moreover it is clear that, for any two derivations $d_1, d_2$ of $\mathcal{P}(\overline \g)$, we have $\widehat{[d_1,d_2]}=[{\hat{d}_1},{\hat{d}_2}]$.
\end{definition}

\begin{lemma} \label{newlemma}
Let $C \in \mathcal{Z}$. The derivation $d/dt_C$ has parity $p(C)$ and
\begin{equation} \label{toprove}
\frac{ d\mathcal{L}_s}{dt_C}=-[M_{z^{-1}C}^+,\mathcal{L}_s].
\end{equation}
\end{lemma}
\begin{proof}
We know that $[M_C,\mathcal{L}_s]=0$, hence $[M_C^+,\mathcal{L}_s]=-[M_C^-,\mathcal{L}_s]$ does not depend on $z$. More precisely we have
\begin{equation}
[M_C^+,\mathcal{L}_s]=-[ M_C^{-1}, s \otimes 1].
\end{equation}
  By Lemma \ref{Lemma: h_C}, 
\begin{equation*}
[M_C^{-1},s \otimes 1]=[(M_C^{-1})_{\mathfrak{b}_-},s \otimes 1]=[(M_{zC}^0)_{\mathfrak{b}_-},s \otimes 1]=[\,  \textstyle  \sum_{t\in I} q_t \otimes \frac{\delta h_{zC}}{\delta \bar{q}_t},s \otimes 1]
\end{equation*}
where we used that $s \in \text{ker} \, \text{ad} \, \, \n \supset \mathfrak{m}$.
Hence, by invariance of the bilinear form, we have for all $t' \in \tilde {I}$
\begin{equation} \label{y}
\begin{split}
\textstyle ([ M_C^{-1},s \otimes 1]|q_{t'}\otimes 1) &=(-1)^{p(C)+1}(s \otimes 1 | [ \textstyle \sum_{t\in I} q_t \otimes \frac{\delta h_{zC}}{\delta \bar{q}_t}, q_{t'} \otimes 1]) \\
&=(-1)^{p(C)+1} \textstyle  \sum_{t\in I} (-1)^{(t+p(C))t'} (s|[q_t,q_{t'}]) \frac{\delta h_{zC}}{\delta \bar{q}_t} \\
 &= \{ \smallint h_{zC}{}_\chi \bar{q}_{t'} \}_2|_{\chi=0}.
\end{split}
\end{equation}
We used equation \eqref{z} to deduce the last equality in \eqref{y}. Note that for $t' \in \tilde I \setminus I$, this quantity is $0$ since $\{ \overline \n {}_\chi \mathcal{P}(\overline \g) \}_2=0$. It follows from \eqref{y} that 
\begin{equation}
[M_C^{-1},s \otimes 1]=\sum_{t \in I} (-1)^{tp(C)} q^t \otimes  \{ \smallint h_{zC}{}_\chi \bar{q}_{t} \}_2|_{\chi=0},
\end{equation}
which concludes the proof.
\end{proof}

\begin{theorem}\label{bigthe}
The super Hamiltonian equations 
\begin{equation} \label{hamiltequa}
\frac{d \bar{q}_t }{dt_C}=\{\smallint h_C {}_\chi \bar{q}_t \}_2 |_{\chi=0}, \, \, t \in \tilde I
\end{equation}
associated with elements $C \in \mathcal{Z}$ are pairwise compatible. In other words, the derivations $d/dt_{C_1}$ and $d/dt_{C_2}$ of $\mathcal{P}(\overline \g)$  supercommute for all $C_1, C_2 \in \mathcal{Z}$. For all $C \in \mathcal{Z}$, $\smallint h_C$ is an integral of motion of each of these equations. 
\end{theorem}
\begin{proof}
We are going to prove that for any $C_1,C_2 \in \mathcal{Z}$, $d/dt_{C_1}$ and $d/dt_{C_2}$ supercommute as derivations of $\mathfrak{g} \otimes \mathcal{P}(\overline \g)$. By Definition \ref{extder}, it  implies that $d/dt_{C_1}$ and $d/dt_{C_2}$ supercommute as derivations of $\mathcal{P}(\overline \g)$ which in turn implies that $\smallint h_{C_1}$ and $\smallint h_{C_2}$ are in involution for the second reduced bracket by Lemma \ref{eqncom} and Lemma \ref{newlemma}.
\\
In the rest of the proof we follow the lines of Lemma $1.6$ in \cite{DS}.
Let $A_1=-M_{z^{-1}C_1}$ and $A_2=-M_{z^{-1}C_2}$. We have
\begin{equation}
\frac{d \mathcal{L}_s}{dt_{C_1}}=[A_1^+, \mathcal{L}_s], \, \, \, \, 
\frac{d \mathcal{L}_s}{dt_{C_2}}=[A_2^+, \mathcal{L}_s].
\end{equation}
Recall that $\mathcal{L}_s=e^{-\text{ad} \, T}(D+K+\Lambda\otimes 1)$. It follows that 
\begin{equation}
\frac{d K}{dt_{C_i}}=[ \tilde{A_i}, D+K+\Lambda \otimes 1] , \  i=1,2, 
\end{equation}
where $\tilde{A_i}=e^{\text{ad}\, T}A_i^+ +\frac{dT}{dt_{C_i}}+\frac{1}{2} [T,\frac{dT}{dt_{C_i}}]+\frac{1}{6} [T,[T,\frac{dT}{dt_{C_i}}]]+\cdots$. By the same argument as in the proof of Lemma \ref{second lemma_DS}, we deduce  that $\tilde{A}_i \in \mathcal{K}\otimes \mathcal{P}$. 
We have 
\begin{equation}
0=\frac{d}{dt_{C_i}}(C_j \otimes 1)=[\tilde{A}_i,C_j \otimes 1], \, \, i,j=1,2,
\end{equation}
therefore 
\begin{equation}
\frac{d A_j}{dt_{C_i}}=[{A_i}^+,A_j], \, \, \, i,j=1,2.
\end{equation}
From here, and using $[A_i,A_j]=0$, it is straighforward to check that
\begin{equation}
[d/dt_{C_1},d/dt_{C_2}](\mathcal{L}_s)=0.
\end{equation}
\end{proof}

\subsection{Integrable super bi-Hamiltonian equations on $\mathcal{W}(\g,\n,\mathfrak{m},f)$}

In the remaining part of this section we will see how the equations \eqref{hamiltequa}, after being reduced to the generalized $W$-algebra
$\mathcal{W}(\g,\n,\mathfrak{m},f)$, give rise to bi-Hamiltonian systems for the two compatible SUSY PVA structures \eqref{bracket_H} and \eqref{secondbra}.

\begin{lemma} \label{Lemma:first_Hamiltonian}
Let $C \in \mathcal{Z}$. For all $\phi\in \mathcal{W}(\g,\n,\mathfrak{m},f)$, we have 
\begin{equation}
\begin{aligned}
\int (M^+_C|[\mathcal{L}_s, \sum_{t\in I} q_t \otimes \frac{\delta \phi}{\delta \bar{q}_t} ] )  =-  \int \{h_C{}_\chi \phi\}_1^{\mathcal{W}}|_{\chi=0}.
\end{aligned}
\end{equation}
\end{lemma}
\begin{proof}
By Proposition \ref{Prop: Lax_bracket_W}, we can see that 
\begin{equation*}
\textstyle \int (\sum_{t\in I} q_t \otimes \frac{\delta h_C}{\delta \bar{q}_t}|[\mathcal{L}_s,\sum_{t'\in I } q_{t'} \otimes \frac{\delta \phi}{\delta \bar{q}_{t'}}] ) =-  \int \{h_C{}_\chi \phi\}_1^{\mathcal{W}}|_{\chi=0}.
\end{equation*}
 Hence we need to show that
\begin{equation}
\begin{aligned}
\textstyle \int (M^0_C|[\mathcal{L}_s, \sum_{t\in I} q_t \otimes \frac{\delta \phi}{\delta \bar{q}_t} ] )  \textstyle =\int (\sum_{t\in I} q_t \otimes \frac{\delta h_C}{\delta \bar{q}_t}|[\mathcal{L}_s,\sum_{t'\in I} q_{t'} \otimes \frac{\delta \phi}{\delta \bar{q}_{t'}}] ).
\end{aligned}
\end{equation}
By Lemma \ref{Lemma: h_C}, it is enough to check that 
\begin{equation} \label{Eqn:5.9_1}
 \textstyle    ((M_C)^0_{\n}|[\mathcal{L}_s, \sum_{t \in I} q_t \otimes \frac{\delta \phi}{\delta \bar{q}_t} ] )=0.
 \end{equation}
It is clear by definition of the bilinear form that
\begin{equation}\label{n_lax=0}
\begin{aligned}
\textstyle   (n\otimes r |[\mathcal{L}_s, \sum_{t\in I} q_t \otimes \frac{\delta \phi}{\delta \bar{q}_t} ] )
& = \textstyle  (-1)^{p(r) p(\phi)}  (n\otimes 1 |[\mathcal{L}_s, \sum_{t\in I} q_t \otimes \frac{\delta \phi}{\delta \bar{q}_t} ] )r 
\end{aligned}
\end{equation}
for all $n \otimes r \in \n \otimes \mathcal{P}$.
For all $n \in \n$, 
since $\{\overline{n}{}_\chi \phi\}_1 \in  \mathcal{I}_{f, \mathfrak{m}}[\chi]$, 
\begin{equation} \label{n_lax=1}
\begin{aligned}
 & \textstyle  (n\otimes 1 |[D+\sum_{t\in \tilde{I}} q^t \otimes \overline{q}_t, \sum_{t'\in I} q_{t'} \otimes \frac{\delta \phi}{\delta \overline{q}_{t'}} ] )\\
& = \textstyle \sum_{t'\in I} (-1)^{p(n)p(t')+p(n)} (\overline{[q_{t'},n]}+(q_{t'}|n)D) \frac{\delta \phi}{\delta\overline{q}_{t'}}  \\
& = (-1)^{p(\bar{n})p(\phi)+1} \{ \phi{}_\chi \bar{n}\}_1|_{\chi=0} \in  \mathcal{I}_{f, \mathfrak{m}}[\chi].
\end{aligned}
\end{equation} 
To obtain the last equality of \eqref{n_lax=1}, we used \eqref{Eqn:mater formula-local}. 

After replacing $\bar{q}_t$ with $(f|q_t)$ for $t\in \tilde{I} \setminus I$, we deduce that 
\begin{equation} \label{n_lax=4}
\textstyle   (n\otimes 1 |[\mathcal{L}_s, \sum_{t\in I} q_t \otimes \frac{\delta \phi}{\delta \bar{q}_t} ] )
 \in \mathcal{I}_{f, \mathfrak{m}} \cap \mathcal{P}=\{0\}.
\end{equation}
Equation \eqref{Eqn:5.9_1} follows from \eqref{n_lax=0} and \eqref{n_lax=4}.
\end{proof}

\begin{lemma}\label{Lemma:second_Hamiltonian}
Let $C \in \mathcal{Z}$. For all $\phi\in \mathcal{W}(\g,\n,\mathfrak{m},f)$, we have 
\begin{equation}\label{Eqn:5.10_1}
\begin{aligned}
\int (M^{-}_{C} |[\mathcal{L}_s, \sum_{t\in I} q_t \otimes \frac{\delta \phi}{\delta \bar{q}_t} ] ) =  \int \{h_{zC}{}_\chi \phi \}_2^{\mathcal{W}} |_{\chi=0}.
\end{aligned}
\end{equation}
\end{lemma}
\begin{proof}
Let $\phi\in \mathcal{W}(\g,\n,\mathfrak{m},f)$.
We know by Proposition \ref{pro} that 
\begin{equation}\label{Eqn:5.10_2}
\begin{aligned}
 \textstyle \int (\sum_{t\in I} q_t \otimes \frac{\delta h_{zC}}{\delta \bar{q}_t}|[s\otimes 1,\sum_{t'\in I} q_{t'} \otimes \frac{\delta \phi}{\delta \bar{q}_{t'}}] )= \textstyle \int \{h_{zC}{}_\chi \phi \}_2^{\mathcal{W}} |_{\chi=0}.
\end{aligned}
\end{equation}
Moreover, since
 $(m \otimes r|[s\otimes 1,\sum_{t\in I} q_t \otimes \frac{\delta \phi}{\delta \bar{q}_t}] )=0$
for any $m \in \mathfrak{m} \subset \n$ and $r \in \mathcal{P}$, we have
\begin{equation} \label{aaa}
\textstyle (M_C^{-1}|[s\otimes 1,\sum_{t\in I} q_t \otimes \frac{\delta \phi}{\delta \bar{q}_t}] )=((M_C^{-1})_{\mathfrak{b}_-}|[s\otimes 1,\sum_{t\in I} q_t \otimes \frac{\delta \phi}{\delta \bar{q}_t}] ).
\end{equation}
Finally, from Lemma \ref{Lemma: h_C} we get that
\begin{equation}\label{bbb}
\textstyle (M_C^{-1})_{\mathfrak{b}_-}=(M_{zC}^0)_{\mathfrak{b}_-}=\sum_{t\in I} q_t \otimes \frac{\delta h_{zC}}{\delta \bar{q}_t}.
\end{equation}
The result follows from equations \eqref{Eqn:5.10_2}, \eqref{aaa} and \eqref{bbb}.
\end{proof}

\begin{theorem} \label{mainth}
 The super Hamiltonian equations
\begin{equation}\label{redSHE}
\frac{d w_t}{dt_C} = \{ \smallint h_{C}{}_\chi  w_t \}_1^{\mathcal{W}}|_{\chi=0}, \, \,t \in J
\end{equation} 
associated to elements $C \in \mathcal{Z}$ are pairwise compatible. In particular, the Hamiltonians $(\int h_{D})_{D \in \mathcal{Z}}$ are  integrals of motion for each of these equations. Moreover, \eqref{redSHE} can be rewritten using the second reduced bracket:
\begin{equation*}
\frac{d w_t}{dt_C} = \{ \smallint h_{zC}{}_\chi w_t \}_2^{\mathcal{W}}|_{\chi=0}, \, \, t \in J.
\end{equation*}
\end{theorem}

\begin{proof}
Since $M_C$ commutes with $\mathcal{L}_s$, it follows immediately from Lemma \ref{Lemma:first_Hamiltonian}, Lemma \ref{Lemma:second_Hamiltonian} and from the invariance of the bilinear form that
\begin{equation}\label{Eqn:thm_biHamil_3}
 \{\smallint h_{C}, \smallint\phi\}_1^{\mathcal{W}}=  \{\smallint h_{zC}, \smallint \phi\}_2^{\mathcal{W}} \, \, \text{for all} \, \, \phi \in \mathcal{W}(\g,\n,\mathfrak{m},f).
\end{equation}
By Corollary \ref{cor_1&2}, we can remove integrals and deduce that
\begin{equation}
\{h_{C}{}_\chi  \phi\}_1^{\mathcal{W}}|_{\chi=0}=\{h_{zC}{}_\chi  \phi\}_2^{\mathcal{W}}|_{\chi=0} \, \, \text{for all} \, \, \phi \in \mathcal{W}(\g,\n,\mathfrak{m},f).
\end{equation}
On the other hand, by definition of the second reduced bracket introduced in Proposition \ref{keyprop} and by Theorem \ref{bigthe} one has
\begin{equation}
\textstyle \{ \int h_C, \int h_D \}_2^{\mathcal{W}}=0, \, \, \, \text{for all} \, \, C,D \in \mathcal{Z}.
\end{equation}
The result now follows from Lemma \ref{eqcom}.
\end{proof}

\section{Example : $\g=\mathfrak{osp}(2|2)$}\label{Sec:example}
\if Recall that $i'+i'' \geq i$ implies that $[f, \mathfrak{n}]$ is a subset of $\mathfrak{b}$. The inequality $d-j \leq i''$ corresponds to $s \in \text{ker} \, \text{ad} \, \mathfrak{n}$ and $d-j <i'$ corresponds to $Im \, ad \, s \subset \mathfrak{b}$. Of course we want $s \in \mathfrak{n}$ so $j > i''$.
\\
\\
\fi
Let $\mathfrak{g}$ be the Lie superalgebra $\mathfrak{osp}(2|2)$. Recall that it is $4+4$ dimensional, with the even space spanned by the elements $f_3,h_1,h_2,e_3$, and the odd space spanned by the elements $e_1,e_2,f_1,f_2$. The relations between these elements are given by 
\begin{equation*}
\begin{split}
[e_1,e_2]&=e_3, \, \, \, \,  \, \, \,\, \,    [f_1,f_2]=f_3, \, \,\, \,  \, \, \, \,[h_2,e_1]=2 e_1, \,  \, [h_1,e_2]=2e_2,\, 
\\ 
[h_1,f_2]&=-2f_2,\,    [h_2,f_1]=-2f_1, \, [e_1,f_1]=h_1, \, \, \, \,[e_2,f_2]=h_2  . 
\end{split}
\end{equation*}
The Lie superalgebra $\mathfrak{g}$ can be represented as a sub Lie superalgebra of $\mathfrak{sl}(2|2)$ with 
\begin{equation*}
e_1=
 \begin{pmatrix} 
0 & 0 &0 &0  \\
0& 0 & 0& 1 \\
1&0 &0 & 0 \\
0 & 0 & 0 & 0 
\end{pmatrix}
e_2=
\begin{pmatrix} 
0 & 0 &0 &1  \\
0& 0 & 0& 0 \\
0&1 &0 & 0 \\
0 & 0 & 0 & 0 
\end{pmatrix}
f_1=
\begin{pmatrix} 
0 & 0 &0 &0  \\
0& 0 & 1& 0 \\
0&0 &0 & 0 \\
-1 & 0 & 0 & 0 
\end{pmatrix}
f_2=
\begin{pmatrix} 
0 & 0 &1 &0  \\
0& 0 & 0& 0 \\
0&0 &0 & 0 \\
0 & -1 & 0 & 0 
\end{pmatrix}
\end{equation*}

\if 
Note that $(\frac{e_3}{2}, \frac{h_1+h_2}{2}, \frac{-f_3}{2})$ is a $sl(2)$ triple. We denote $\frac{h_1+h_2}{2}$ by $h$.
\fi

The only non trivial pairings via the super Killing form $( \ | \ )$ are given by 
\begin{equation*}
(e_1 |f_1)=-2, \, \, (e_2|f_2)=-2, \, \, (e_3|f_3)=4, \, \, (h_1|h_2)=-4.
\end{equation*}
Let us consider the following grading on $\mathfrak{g}$:
\begin{equation} \label{gr1}
\mathfrak{g}_{1}=\text{Span}_\C\{e_2, e_3 \}, \quad  \mathfrak{g}_0= \text{Span}_\C\{ f_1,  h_1, h_2,  e_1 \}, \quad \mathfrak{g}_{-1}=\text{Span}_\C\{ f_2,  f_3 \}
\end{equation} 
and let
\begin{equation}
\mathfrak{m}=\mathfrak{n}=\mathfrak{g}_1, \quad \mathfrak{b}=\mathfrak{g}_1 \oplus \mathfrak{g}_0, \quad \mathfrak{b}_-=\mathfrak{g}_{-1} \oplus \mathfrak{g}_0,\quad f=f_2, \quad s= e_2.
\end{equation} 
Using the notations in the previous sections, we can check that  the assumptions (A-1), $\cdots$, (A-5) hold.

\subsection{SUSY PVA structures on $\mathcal{W}(\g,\n,\mathfrak{m},f)$} 

Let us find the generators of $\mathcal{W}(\g,\n,\mathfrak{m},f)$. Recall that
 \[\mathcal{P}= S(\C[D]\otimes \bar{\mathfrak{b}}_{-}),\]
where $ \bar{\mathfrak{b}}_{-}$ is spanned by the even elements  $\bar{f}_2, \bar{f}_1, \bar{e}_1$ and the odd elements $\bar{h}_1, \bar{h}_2, \bar{f}_3$. The \textit{universal Lax operator} associated to the given data is 
\begin{equation*}
\mathcal{L}_u=D+\mathcal{Q}_u+f_2 \otimes 1,
\end{equation*} 
where
\begin{equation*}
\begin{split}
 &  \textstyle \mathcal{Q}_u=e_3 \otimes \frac{\bar{f}_3}{4}- e_2 \otimes \frac{\bar{f}_2}{2}- e_1 \otimes \frac{\bar{f}_1}{2}- h_1 \otimes \frac{\bar{h}_2}{4}- h_2 \otimes \frac{\bar{h}_1}{4}+ f_1 \otimes \frac{\bar{e}_1}{2}.
 \end{split}
 \end{equation*}
    \if 
   One should think of $\mathcal{L}_s$ as an element of $\mathfrak{K} \ltimes (\hat{\g} \otimes \mathcal{P}) $. 
   To make $\mathcal{L}_s$ homogeneous of degree $-1$, we let $D, \bar{f_1}, \bar{h_1}, \bar{h_2}$ and $\bar{e_1}$ have grading $-1$ and $z, \bar{f_2}$ and $\bar{f_3}$ have grading $-2$. 
   \fi


Let us take the subspace $V= \text{Span}_{\C} \{ f_1, h_1, e_2, e_3 \} \subset \g$. Then $\mathfrak{b}= [f_2, \n] \oplus V$. Hence, by Lemma \ref{Lem:Lax_can}, there exists a unique $N \in (\mathfrak{n} \otimes \mathcal{P})_{\bar{0}}$ such that 
 \begin{equation*}
 e^{ad N} \mathcal{L}_u=D+\mathcal{Q}^c + f \otimes 1:= \mathcal{L}^{c},
 \end{equation*}
 where $\mathcal{Q}^c$ has the form $e_3 \otimes w_4+ e_2 \otimes w_3+h_1 \otimes w_2+ f_1 \otimes w_1$.  A direct computation gives
 \begin{equation*}
 \begin{split}
& \textstyle  N=-e_2 \otimes \frac{\bar{h}_1}{4}- e_3 \otimes \frac{\bar{f}_1}{4},\\ 
&  \textstyle w_1=\frac{\bar{e}_1}{2}, \, \, w_2=-\frac{\bar{h}_2}{4}, \, \,  w_3=-\frac{\bar{f}_2}{2} -\frac{\bar{h}'_1}{4}+ \frac{\bar{e}_1\bar{f}_1}{4}-\frac{\bar{h}_1\bar{h_2}}{8},\,  \, w_4=\frac{\bar{f}_3}{4}+\frac{\bar{f}'_1}{4}-\frac{\bar{f}_1 (\bar{h}_1+\bar{h}_2)}{8}.
\end{split}
\end{equation*}
By  Theorem \ref{Lem:equi_W_1}, the generalized $W$-algebra associated to $\g, f, \mathfrak{m}$ and $\n$ is 
\[ \mathcal{W}(\g,\n,\mathfrak{m},f)= S(\C[D]\otimes \text{Span}_\C\{ w_1, w_2, w_3, w_4\})\subset \mathcal{P}.\]

\if
 It is the subalgebra of invariant functions on $\mathcal{P}$ under the action of  $(\mathfrak{n} \otimes \mathcal{P})_0$. More precisely, for all $X  \in (\mathfrak{n} \otimes \mathcal{P})_{\bar{0}}$ of degree $0$, we let $\bar{f_3}^X,\cdots, \bar{e_1}^X$ be the elements of $\mathcal{P}$ such that 
\begin{equation*}
e^{ad X} \mathcal{L}_s= D+ e_3 \otimes \frac{\bar{f_3}^X}{4}-\cdots+f_1 \otimes \frac{\bar{e}_1^X}{2}+ \Lambda \otimes 1
 \end{equation*}
 and naturally extend the action of $X$ to $\mathcal{P}$. Then, by Theorem \ref{Lem:equi_W_1}, an element $p$ in $\mathcal{P}$ is such that $p^X=p$ for all $X \in (\mathfrak{n} \otimes \mathcal{P})_0$ if and only if $p \in \mathcal{W}$. 
 \fi
 
 By Proposition \ref{keyprop}, $\mathcal{W}(\g,\n,\mathfrak{m},f)$ admits two compatible SUSY PVA structures, induced from the compatible affine SUSY PVA structures on $S(\mathbb{C}[D] \otimes \bar{ \mathfrak{g}})$.  We give the $\chi$-brackets between the generators $w_i$ for the first structure:
\begin{equation*}
\begin{split}
\{{w_1}_{\chi} w_2 \}_{1}&=-\frac{1}{2}w_1, \, \, \, \{{w_1}_{\chi} w_3 \}_{1}=\frac{1}{2} \chi w_1, \, \, \, \{{w_1}_{\chi} w_4 \}_{1}=-\frac{1}{2}w_3+\frac{1}{2} \chi w_2 - \frac{1}{4} \chi^2.
\\
\{{w_2}_{\chi} w_3 \}_{1}&=\frac{1}{2} \chi w_2+ \frac{1}{4} \chi^2, \, \, \, \{{w_2}_{\chi} w_4 \}_{1}=\frac{1}{2} w_4, \, \, \,\{{w_3}_{\chi} w_4 \}_{1}=-\frac{1}{2} w_4'-2 w_2 w_4,\\
\{{w_3}_{\chi} w_3 \}_{1}&=-\frac{1}{2}w_3'+2w_1w_4-2w_2w_3, \, \, \, \, \{{w_1}_{\chi} w_1 \}_{1}=\{{w_2}_{\chi} w_2 \}_{1}=\{{w_4}_{\chi} w_4 \}_{1}=0.
\end{split}
\end{equation*}
As for the second SUSY PVA structure, it is defined by the two $\chi$-brackets
\begin{equation*}
\{{w_1}_{\chi} w_4 \}_{2}=\frac{1}{2}, \, \, \, \{{w_3}_{\chi} w_3 \}_{2}=2w_2,
\end{equation*}
and all other $\chi$-brackets on pairs of generators being trivial.

\subsection{Super bi-Hamiltonian system associated to $\mathcal{W}(\g,\n,\mathfrak{m},f)$} \ 

Since  $\Lambda^2=z h_2$ is a semisimple element in $\hat{\g}$, we have  $\hat{\g}= \text{ker} \, \text{ad} \, \Lambda^2\oplus \text{im} \, \text{ad} \, \Lambda^2$.  By a direct computation,  we get 
\begin{itemize}
\item $\mathcal{K}:= \text{ker} \, \text{ad} \, \Lambda^2= \C(\!( z^{-1})\!) \otimes \text{Span}_\C\{ f_2, \, h_1, \, h_2, \, e_2\}$, 
\item $ \mathcal{I}:=\text{im} \, \text{ad} \, \Lambda^2= \C(\!( z^{-1})\!) \otimes \text{Span}_\C\{ f_3, \, f_1, \, e_1, \, e_3\}$, 
\item $\mathcal{Z}(\mathcal{K})=\C(\!( z^{-1})\!) \otimes \text{Span}_\C\{ h_2\}$.
\end{itemize}


\vspace{1 mm}

Let  $(T^c,H^c) \in (\mathcal{I} \otimes \mathcal{W}(\g,\n,\mathfrak{m},f) \times (\mathcal{K} \otimes \mathcal{W}(\g,\n,\mathfrak{m},f) $ as in Lemma \ref{KW} be such that

\if
Recall that the element  $K^c\in \mathcal{K} \otimes \mathcal{W}(\g,\n,\mathfrak{m},f)$ in Lemma \ref{KW} can be computed by the unique element $T^c=T_{\mathcal{Q}_u^c}\in \mathcal{I} \otimes \mathcal{W}(\g,\n,\mathfrak{m},f)$ in Lemma \ref{first lemma_DS} and \eqref{TH2} in Lemma  \ref{second lemma_DS}: 
\fi
\begin{equation} 
e^{ad \, T^c}\mathcal{L}^c=D+K^c+\Lambda \otimes 1. 
\end{equation}

\begin{theorem}
For all $n \geq 0$, let 
\begin{equation}
\smallint \rho_n=\smallint (K^c|h_2 \otimes z^n).
\end{equation}
The super hamiltonian equations
\begin{equation} \label{hierarchy}
\frac{dw_t}{dt_n} = \{\smallint \rho_n{}_\chi  w_t \}_1|_{\chi=0}, \, \, t=1,2,3,4
\end{equation} 
are  pairwise compatible. Moreover, the hamiltonians $\int \rho_k,\,  k \geq 0$ are integrals of motion for each of these equations, and linearly independent over $\mathbb{C}$. Finally, for all $n \geq 0$,  the $n$-th system of equation in the hierarchy can be rewritten using the second bracket:
\begin{equation}\label{hierarchy2}
\frac{d w_t}{dt_n} = \{ \smallint \rho_{n+1}{}_\chi  w_t \}_2|_{\chi=0}, \, \, t=1,2,3,4.
\end{equation}
\end{theorem}

\begin{proof}
By Theorem \ref{mainth}, the equations \eqref{hierarchy} are pairwise compatible super Hamiltonian system with integrals of motion $\int \rho_k,\,  k \geq 0$. 

Let us show that  $\smallint \rho_n$ are linearly independent hamiltonians using Proposition \ref{indept}. In order to compute the linear part of $T^c$, it is easier to use the equation 
\begin{equation} \label{eq111}
e^{ad \, T^c} [\mathcal{L}^c, \mathcal{L}^c]/2=D^2+H^c+\Lambda^2 \otimes 1.
\end{equation}
which can be found in Lemma \ref{first lemma_DS}.
We have 
\begin{equation*}
[\mathcal{L}^c,\mathcal{L}^c]/2=D^2+\mathcal{U}+\Lambda^2 \otimes 1
\end{equation*}
 with $\mathcal{U}$ breaking into three homogeneous parts
\begin{equation*}
\begin{split}
& \mathcal{U}_1=e_3 \otimes [w_4'+2w_2w_4]+e_2 \otimes [-w_3'+2w_1w_4-2w_2w_3],\\
& \mathcal{U}_0=2 e_1 \otimes w_4 +h_1 \otimes w_2' + h_2 \otimes w_3 - f_1 \otimes w_1' , \\
& \mathcal{U}_{-1} = -2z e_2 \otimes w_2+2 f_2 \otimes w_2 +f_3 \otimes w_1.
\end{split}
\end{equation*}
Hence the linear part of $\mathcal{U}_{\mathcal{I}}$ is
\begin{equation*}
(\mathcal{U}_{\mathcal{I}})_l=e_3 \otimes w_4'+2 e_1 \otimes w_4-f_1 \otimes w_1'+f_3 \otimes w_1.
\end{equation*}
The linear component of equation \eqref{eq111} projected on $\mathcal{I}$ reads
\begin{equation*}
T_l''+[h_2 \otimes z,T_l]=(\mathcal{U}_{\mathcal{I}})_l.
\end{equation*}
In terms of coordinates, $T_l=e_3 \otimes A+e_1 \otimes B+f_1 \otimes C +f_3 \otimes D$ with even elements $A,D$ and odd elements $B,C$. This equation is equivalent to the four ODEs
\begin{equation*}
\begin{split}
A''+2zA&=w_4', \, \, \,  \, \, \,C''-2zC=-w_1',\\
 B''+2zB&=2w_4, \, \, \,  D''-2zD=w_1.
\end{split}
\end{equation*}
We solve these equations for $B$ and $D$:
\begin{equation*}
\textstyle B=2\sum_{k \geq 1} (-1)^{k-1}\frac{w_4^{(2k-2)}}{(2z)^k}, \, \, \, \, \, \, \, \, \, D=-\sum_{k \geq 1}\frac{w_1^{(2k-2)}}{(2z)^k}.
\end{equation*}
By Proposition \ref{indept}, the quadratic part of $\int \rho_n$ is given by the expression
\begin{equation*}
\begin{split}
\smallint (\rho_n)_{quad}&=\frac{1}{2} \smallint (e_3 \otimes w_4+f_1 \otimes w_1| z^n(2 e_3 \otimes A+2 e_1 \otimes B-2 f_1 \otimes C -2 f_3 \otimes D))\\
&=\smallint -4w_4D_n+2w_1B_n \\
&=4 \smallint \frac{w_4w_1^{(2n-2)}}{2^n}+(-1)^{n-1}\frac{w_1w_4^{(2n-2)}}{2^n}\\
&=2^{3-n}\smallint w_1^{(2n-2)}w_4.
\end{split}
\end{equation*}
\end{proof}

\begin{proposition}
 The first two equations in the hierarchy \eqref{hierarchy} are given by
\begin{equation}\label{ex_eq1}
\frac{dw_1}{dt_0}=2w_1, \, \, \, \frac{dw_2}{dt_0}=0, \, \, \, \frac{dw_3}{dt_0}=0, \, \, \,\frac{dw_4}{dt_0}=-2w_4,
\end{equation}
and
\begin{equation}\label{ex_eq1}
\begin{split}
\frac{dw_1}{dt_1}&=w_1''-2(w_1w_2)'-2w_1w_3, \, \, \, \frac{dw_2}{dt_1}=0, \\
\frac{dw_3}{dt_1}&=-8 w_1 w_2 w_4, \, \, \,\frac{dw_4}{dt_1}=w_4''+2w_2w_4'+2w_3w_4.
\end{split}
\end{equation}
The three first integrals of motion of the hierarchy \eqref{hierarchy} are
\begin{equation*}
\smallint \rho_0=-\smallint 4 w_2, \, \, \, \smallint \rho_1= 4 \smallint w_1w_4, \, \, \, \smallint \rho_2= \smallint 2w_1''w_4-4(w_1w_2)'w_4-4w_1w_3w_4,
\end{equation*}

\end{proposition}
\begin{proof}
Direct computation.
\end{proof}

\end{document}